\documentclass{article}

\usepackage[utf8]{inputenc}
\usepackage[english]{babel}
\usepackage[T1]{fontenc}
\usepackage{csquotes}
\usepackage[dvipsnames]{xcolor}

\usepackage{amsfonts,amsthm,amsmath,amssymb,mathtools}
\usepackage{algorithm, algpseudocode}
\usepackage{dsfont,newtxtext,newtxmath,microtype}

\usepackage{graphicx, quantikz, tikz,tikz-3dplot}
\usepackage{authblk}
\usepackage{booktabs,multirow}
\usepackage{braket}
\usepackage{cancel}
\usepackage[font=small]{caption}
\usepackage{comment}
\usepackage{diagbox}
\usepackage{enumerate}
\usepackage{float}
\usepackage{fullpage}
\usepackage[section]{placeins}
\usepackage{quantikz}
\usepackage{subfigure}
\usepackage{slashed}
\usepackage{todonotes}
\usepackage{bbm} 
\usepackage{yfonts}
\usepackage{algorithm}
\usepackage{algpseudocode}

\pgfdeclarelayer{background}
\pgfsetlayers{background,main}

\definecolor{navyblue}{rgb}{0.0, 0.0, 0.5}
\definecolor{LightPink}{rgb}{0.858, 0.188, 0.478}

\usepackage[colorlinks=true, urlcolor=blue,citecolor=purple,anchorcolor=blue, linkcolor=LightPink]{hyperref}

\newcommand{\epsint}{\epsilon_\mathrm{ext}}

\newcommand{\epsdata}{\epsilon_\mathrm{data}}

\newcommand{\cT}{\mathcal{T}}

\newcommand{\cG}{\mathcal{G}}

\newcommand{\cE}{\mathcal{E}}
\newcommand{\cI}{\mathcal{I}}

\newcommand{\obsv}{A}
\newcommand{\tobsv}{\tilde{\obsv}}
\newcommand{\state}{\rho}
\newcommand{\tstate}{\tilde{\state}}

\newcommand{\estimator}{\hat{\obsv}}
\newcommand{\testimator}{\tilde{\hat{\obsv}}}

\DeclarePairedDelimiter\norm{\lVert}{\rVert}
\DeclarePairedDelimiter\dnorm{\lVert}{\rVert_\diamond}
\DeclarePairedDelimiter\expval{\langle}{\rangle}

\DeclareMathOperator{\ad}{ad}
\DeclareMathOperator{\tr}{tr}
\DeclareMathOperator{\polylog}{polylog}

\usepackage[capitalise]{cleveref}
\Crefname{lemma}{Lemma}{Lemmas}
\Crefname{proposition}{Proposition}{Propositions}
\Crefname{definition}{Definition}{Definitions}
\Crefname{theorem}{Theorem}{Theorems}
\Crefname{conjecture}{Conjecture}{Conjectures}
\Crefname{corollary}{Corollary}{Corollaries}
\Crefname{example}{Example}{Examples}
\Crefname{section}{Section}{Sections}
\Crefname{appendix}{Appendix}{Appendices}
\Crefname{figure}{Fig.}{Figs.}
\Crefname{equation}{Eq.}{Eqs.}
\Crefname{table}{Table}{Tables}
\Crefname{item}{Property}{Properties}
\Crefname{remark}{Remark}{Remarks}
\Crefname{fact}{Fact}{Facts}

\newtheorem{theorem}{Theorem}

\newtheorem{lemma}[theorem]{Lemma}

\usepackage{ltablex}

\newcommand\prob\textsc
\makeatletter
\newcommand{\probleminput}[1]{\gdef\@probleminput{#1}}
\newcommand{\problemquestion}[1]{\gdef\@problemquestion{#1}}
\newcommand{\problempromise}[1]{\gdef\@problempromise{#1}}
\DeclareDocumentEnvironment{problem}{}{
\probleminput{}\problempromise{}\problemquestion{}
\leavevmode
}{
  \par\addvspace{0\baselineskip}
  \noindent
  \begin{itemize}
      \item[\textbf{Input:}] \@probleminput
\ifx\@problempromise\empty%
\else%
      \item[\textbf{Promise:}] \@problempromise
\fi%
      \item[\textbf{Output:}] \@problemquestion
  \end{itemize}
}
\makeatother

\usepackage[backend=biber,style=alphabetic,doi=false,isbn=false,url=false,maxbibnames=7,maxcitenames=2]{biblatex}
\bibliography{references}
\newbibmacro{string+doi}[1]{\iffieldundef{doi}{#1}{\href{https://dx.doi.org/\thefield{doi}}{#1}}}
\DeclareFieldFormat{title}{\usebibmacro{string+doi}{\mkbibemph{#1}}}
\DeclareFieldFormat[article]{title}{\usebibmacro{string+doi}{\mkbibquote{#1}}}
\DeclareFieldFormat[incollection]{title}{\usebibmacro{string+doi}{\mkbibquote{#1}}}
\title{Randomly Compiled Quantum Simulation with Exponentially Reduced Circuit Depths}

\author[1,2]{\href{https://orcid.org/0000-0002-6077-4898}{James~D.~Watson}}

\affil[1]{Joint Center for Quantum Information \& Computer Science, National Institute of Standards and Technology and University of Maryland, College Park}
\affil[2]{Department of Computer Science and Institute for Advanced Computer Studies, University of Maryland, College Park}

\date{}

\begin{document}

{\begingroup
		\hypersetup{urlcolor=navyblue}
\maketitle
		\endgroup}

\begin{abstract}
    
The quantum stochastic drift protocol, also known as qDRIFT, has become a popular algorithm for implementing time-evolution of quantum systems.
    In this work we develop qFLO, a higher order randomised algorithm for time-evolution.
    To estimate an observable expectation value at time $T$ to precision $\epsilon$, we show it is sufficient to use circuit depths of $O(T^2\log(1/\epsilon))$ -- an exponential improvement over standard qDRIFT requirements with respect to $\epsilon$.
    The protocol achieves this using $O(1/\epsilon^2)$ repeated runs of the standard qDRIFT protocol combined with classical post-processing in the form of Richardson extrapolation.
    Notably, it requires no ancillary qubits or additional control gates making it especially promising for near-term quantum devices.
    Furthermore, it is well-conditioned and inherits many desirable properties of randomly compiled simulation methods, including circuit depths that do not explicitly depend on the number of terms in the Hamiltonian, making it an attractive algorithm for problems in quantum chemistry. 
\end{abstract}

\section{Introduction}

\yinipar{S}imulating the dynamics of a quantum system evolving under a Hamiltonian has been one of the key motivations for developing quantum computers, and believed to be a task that quantum computers achieve an exponential speed up relative to classical computers.
Dynamical simulation is key to a wide range of problems in physics including high energy and nuclear physics, quantum chemistry, and material science  \cite{jordan2012quantum, cao2019quantum,shaw2020quantum,roggero2020quantum, watson2023quantum, clinton2024towards}.
Beyond academic interest, it is believed that quantum simulation algorithms will find a wide range of uses in applying quantum computers to design catalysts, batteries, and other materials.

With these use-cases in mind, a wide range of quantum simulation algorithms have been designed including product formulae \cite{lloyd1996universal}, linear combination of unitaries (LCU) \cite{childs2012hamiltonian, low2019well}, and qubitization \cite{low2019hamiltonian}.
Product formulae and its variants are among the most popular, in part due to their good empirical performance, ease of implementation, and conceptual simplicity.
Although the product formula approach by itself is asymptotically worse than LCU or qubitization methods, in practice they often perform similarly or better than their LCU and qubitization counter-partners in the regimes of practical interest \cite{babbush2015chemical,childs2018toward}, and often take advantage of the physics inherent to the system they are simulating ~\cite{tran2020destructive,childs2021theory, tran2021faster, csahinouglu2021hamiltonian, zhao2022hamiltonian, zhao2024entanglement}.
Indeed, other non-Trotter algorithms have been shown to be unable to take advantage of the physics of the system in general \cite{zlokapa2024hamiltonian}.
Furthermore, the performance of these product formulae can be improved using classical post-processing if we are only interested in the expectation values of time-evolved observables rather than the full state \cite{rendon2024improved,watson2024exponentially, cook2024parametric}.
With these benefits in mind, a wide range of optimised product formulae have been developed\footnote{See for example Refs. \cite{yuan2019theory,campbell2019random,ouyang2020compilation,morales2022greatly, nakaji2023qswift,sharma2024hamiltonian, bosse2024efficient, chen2024adaptive, bagherimehrab2024faster}} and many methods for optimising the algorithm at the circuit level have been developed \cite{mckeever2023classically,kang2023leveraging}.

One adaptation of product formula-type methods is the quantum stochastic drift protocol (qDRIFT), which probabilistically implements a unitary evolution with probabilities proportional to the strength of the individual terms in the Hamiltonian \cite{campbell2019random}.
This is a so-called ``randomly compiled'' product formula known as opposed to deterministically compiled product formulae such as Trotter-Suzuki.
qDRIFT has many appealing features, including the fact that each time-step of the algorithm only requires implementing a number of exponential terms which are independent of the number of terms in the Hamiltonian.
This can be compared to Trotter-Suzuki product formulae, where each time-step requires an implementation of a number of terms directly proportional to the number of terms in the Hamiltonian's decomposition.
As such, qDRIFT has become particularly popular for quantum chemistry applications where, for a system with $n$ fermionic modes, there are often $O(n^4)$ of local terms in the Hamiltonian.

 In this work, we present qFLO, a higher-order randomly compiled algorithm for Hamiltonian simulation.
 qFLO is based on the qDRIFT protocol, but makes use of classical post-processing in the form of Richardson extrapolation to improve the error scaling while maintaining the benefits of qDRIFT.
 In particular, we show that varying the size of the time-steps utilised in the qDRIFT approach generates a flow towards the zero step-size limit, which corresponds to the non-Trotterized evolution. 
 By taking measurements of expectation values at finite step-sizes, we can then extrapolate along this flow using a Richardson extrapolation approach to get a better estimate for the zero step-size limit. 
 A similar approach has been shown to be successful in reducing error for staged product formulae, and here we show it can also be applied to randomly compiled formulae \cite{endo2019,watson2024exponentially}. 
 In particular, we achieve an exponentially better scaling in the circuit depths compared to regular qDRIFT, with the maximum circuit depth required scaling as 
 \begin{align}
     O\left( (\lambda T)^2\log(1/\epsilon) \right)
 \end{align}
 compared to $ O\left( \frac{(\lambda T)^2}{\epsilon} \right)$ for the regular qDRIFT protocol.
 Here, if the Hamiltonian is defined as $H=\sum_j h_j H_j$, $\norm{H_j}=1$ and $h_j>0$, then the parameter $\lambda$ is defined as $\lambda = \sum_jh_j$.
 Notably our approach requires no auxiliary qubits or other quantum resources beyond multiple runs of the qDRIFT channel.
 Thus, this work demonstrates that the method of extrapolating via classical post-processing can be successfully applied beyond staged product formulae such as Trotter-Suzuki.

 The structure of the paper is as follows. 
 In \cref{Sec:Overview} we give an outline of the proof methods and the algorithm.
 In \cref{Sec:Results} we discuss our results, the algorithm, and previous literature.
 After this, we proceed with the formal proof of correctness of the algorithm. \cref{Sec:Series_Expansion} gives a series expansion for the time-evolved observable under the iterated qDRIFT channel, \cref{Sec:Richarson_Extrapolation} gives the Richardson extrapolation procedure, \cref{Sec:Richardson_Bounds} gives error bounds on the Richardson estimator, and finally \cref{Sec:Robust_Richardson} shows that the Richardson estimator is robust to measurement error and gives circuit depth and gate count scaling.

\section{Overview of Methods}\label{Sec:Overview}

\paragraph{The qDRIFT Channel.}
We start by reviewing the qDRIFT channel as defined in Ref.~\cite{campbell2019random}.
Given a time-independent local Hamiltonian we can decompose it as $H = \sum_j h_j H_j,$ where $\norm{H_j} =  1$ and $h_j\in \mathbb{R}$. 
The decomposition is chosen such that $h_j>0$ and such that $e^{-iH_jt}$ is implementable for arbitrary times $t$.
Suppose we want to find the time evolution after time $T$. 
We define the qDRIFT channel as the channel
\begin{align}
    \cE(\rho) = \sum_j p_j e^{-i\lambda H_jt}\rho e^{i\lambda H_j t}
\end{align}
for a time-step $t$, where $p_j = \frac{h_j}{\lambda}$ and $\lambda = \sum_j h_j$. 
\citeauthor{campbell2019random} \cite{campbell2019random} shows that by expanding this channel to first order in $t$, one can derive the following error bounds
\begin{align}
     \cE(\rho) &= \rho - i\sum_j h_jt [H_j,\rho] + O(t^2) \\ 
     &= \rho - it [H,\rho] + O(t^2) \\
    \implies \dnorm{\cE(\rho)-e^{-iH   t}\rho e^{iH  t}} &= O(t^2). 
\end{align}

\noindent In particular, after $N$ applications of the qDRIFT channel and choosing a time-scale $t =  T/N$, \citeauthor{campbell2019random} uses submultiplicativity of norms and the triangle inequality to show that
\begin{align}\label{Eq:qDRIFT_Performance}
    \dnorm{\cE^N(\rho) - e^{-iH T}\rho e^{iHT}}= O\left(\frac{(\lambda T)^2}{N}\right).
\end{align}
Thus, by choosing $N=O((\lambda T)^2/\epsilon)$ to be sufficiently large, we can reduce the error $\epsilon$ to be arbitrarily small.

\paragraph{A Series Expansion of the Repeated qDRIFT Channel.}
\Cref{Eq:qDRIFT_Performance} gives error bounds on the performance of the qDRIFT channel.
In this work, we will improve on the error scaling for observables by first deriving an expansion for the repeated channel in terms of the parameter $1/N,$ and use this expansion to gain more information about the output.

In particular, we show that provided $2t <1/\lambda $, one can write the channel in the form of an exponential
\begin{align}
    \cE =: e^{-it \cG(t)},
\end{align}
where $\cG(t)$ is defined by this relation.
For the full channel evolution after $N$ applications and $t = T/N$, this becomes
\begin{align}
    \cE^N &= e^{-iNt \cG(t)} \\
    &= e^{-iT \cG(t)}.
\end{align}
For convenience, now we change variables and choose to work with the inverse step-size $s= 1/N$ so that (in an abuse of notation) $\cG = \cG(s)$ is now a function of $s$.
We can then expand the channel as a series in $s$ of the form
\begin{align}
    \cE^{1/s}(\rho) &= e^{-iT \cG(s)}(\rho) \\
    &= e^{-iHT}\rho e^{iHT} + \sum_{j=1}^\infty s^j \tilde{E}_{j+1,K}(\rho) + F_K(s,T)(\rho),
\end{align}
where $\tilde{E}_{j+1,K}$ are a set of superoperator coefficients and $F_K$ is a superoperator function, with $\dnorm{F_K}=O(s^K)$.
This allows us to write the time-evolved expectation value of the observable $\obsv$ under the qDRIFT channel in a similar series expansion
\begin{align}
    \tr\left[\obsv \  \cE^{1/s}(\rho)\right] = \tr\left[\obsv e^{-iHT}\rho e^{iHT}\right] +  \sum_{j=1}^\infty a_j s^j + h_K(s,T),
\end{align}
for some coefficients $a_j$ and function $h_K(s,T)=O(s^K)$.
Importantly, in the limit of an infinite number of time-steps,  $s\rightarrow 0$, we recover the exact time-evolution.
We rigorously derive the series in \cref{Sec:Series_Expansion}.

\paragraph{Richardson Extrapolation.}


So far we have shown that there exists a series expansion of the time-evolved local observable in terms of the inverse step-size of the qDRIFT channel.
Here we introduce Richardson extrapolation which we will apply to our series expansion to obtain better estimates for time-evolved expectation values \cite{richardson1911approximate, sidi_2003}.

The Richardson extrapolation procedure works as follows. 
Suppose we have a function $g(x)$ for which we wish to find $g(0)$ using samples of $g(x)$. 
Further suppose $g(x)$ has a series expansion 
\begin{align}
    g(x) = g(0) +  a_1 x + a_2 x^2 + \sum_{n=3}^\infty a_n x^n.
\end{align}
Suppose we can evaluate the function at an initial point $x_1$, we then have
\begin{align}
    g(x_1) = g(0) +  a_1 x_1 + a_2 x_1^2 + \sum_{n=3}^\infty a_n x_1^n.
\end{align}
Now choose a second point $x_2=x_1/k_2$
\begin{align}
    g(x_2) = g(0) +  a_1 \frac{x_1}{k_2} + a_2 \frac{x_1^2}{k_2^2} + \sum_{n=3}^\infty a_n \left(\frac{x_1}{k_2}\right)^n.
\end{align}
We can then compute an estimator $\hat{F}_1$ for $g(0)$ by cancelling the first-order terms as
\begin{align}
    \hat{F}_1(x_1) &=\frac{1}{1-k_2}\left( g(x_1) - k_2g(x_2)\right) \\
    &= g(0) + \frac{a_2x_1^2}{1-k_2}\left(1 - \frac{1}{k_2}\right) +  \frac{1}{{1-k_2}}\sum_{n=3}a_nx_1^n\left(1 - \frac{1}{k_2^{n-1}}    \right) \\
    &= g(0) + O(x_1^2).
\end{align}
In a similar manner, using more samples at $x_1/k_3, x_1/k_4,\dots $ we can construct higher order estimators.
Roughly speaking, using $m$ sampled points will allow us to construct an estimator $\hat{F}_m(x_1) = \sum_{j=1}^m b_k g(x_1/k_j)$ such that 
\begin{align}
    |\hat{F}_m(x_1) - g(0)| =O(x_1^{m-1}). 
\end{align}
We give rigorous results for the Richardson extrapolation procedure in \cref{Sec:Richarson_Extrapolation}.

\paragraph{Combining the Series Expansion and Richardson Extrapolation.}

We now apply the Richardson extrapolation scheme to the series expansion for the time-evolved observable.
Defining the function $f_A(s) =  \tr[\obsv \cE^{1/s}(\rho)],$ we that we can construct an estimator 
\begin{align}
    \estimator_m(s) = \sum_{j=1}^m b_j f_A(s/k_j).
\end{align}
To find the values $f_A(s/k_j)$ needed to construct the estimator, we can run the qDRIFT channel with the appropriate time-step size  $t_j = sT/k_j$ and take measurements.
Then, by carefully upper-bounding the coefficients in the series for $f_A(s)$, we see that the error of the order-$m$ Richardson estimator scales as 
\begin{align}
    \left|\estimator_m(s) - \expval{A(T)}\right| = O\left( s^m (\lambda T)^{2m} \right).
\end{align}
From this we can show that it is sufficient to choose $m=O(\log(1/\epsilon))$ and $\max_i \frac{k_j}{s} = \max_i N_i = O\left((\lambda T)^2 \log(1/\epsilon)\right)$.
We given a rigorous analysis of this in \cref{Sec:Richardson_Bounds}.

Finally, we realise that each of the measurements used to compute $f_A(s/k_j)$ will have some measurement error associated with it, and this will likely introduce error to our Richardson estimator.
To get error $\epsilon_\text{measure},$ for each measurement value, we will need $O(1/\epsilon^2_\text{measure})$-many samples, hence for each value of $N_i$ we need to run quadratically many times in the error parameter we want to achieve.
Moreover, we need to ensure the Richardson estimator is well conditioned --- i.e. it is sufficiently robust so that errors in the measurements do not ``blow up'' exponentially.
We give an analysis of the measurement error in \cref{Sec:Robust_Richardson}.

The result is that we need to choose $m=O(\log(1/\epsilon))$-many sample points, each of which we need to run a circuit of depth no more than $O((\lambda T)^2 \log(1/\epsilon))$ for a number of times scaling as $O(1/\epsilon^2)$.
This allows us to construct an estimator using only classical post-processing of these measurements.

 \section{Results}\label{Sec:Results}
\begin{samepage}
  \begin{theorem}[qFLO Resource Scaling]
  \label{Theorem:Main_Theorem}
     Suppose we wish to obtain the expectation value of an observable $\obsv$ after evolving an initial state $\rho_0$ for time $T$ under a Hamiltonian $H=\sum_{i=1}^L h_j H_j$ with $\norm{H_j}=1$.
     By varying the time-step size of the qDRIFT mapping and taking measurements at $m$ different values of the time-step, with high probability it is possible to construct an order-$m$ Richardson estimator $\hat{\obsv}$ satisfying
     \begin{align*}
         \left|\hat{\obsv} - \tr\left[ \rho_0 e^{iHT}\obsv e^{-iHT} \right] \right|\leq \epsilon \norm{A} 
     \end{align*}
     using $m=\tilde{O}(\log(1/\epsilon))$ sampling points and circuit depths of no more than 
     \begin{align*}
         O\left((\lambda T )^2\log\left(\frac{1}{\epsilon} \right)\right) 
     \end{align*}
     where $\lambda = \sum_j h_j$. 
     Each sample point used to construct the estimator requires $O(1/\epsilon^2)$ runs of the qDRIFT mapping.
     As a result the total gate count scales as $O\left(\frac{1}{\epsilon^2}(\lambda T )^2\log^2\left(\frac{1}{\epsilon} \right)\right) $.
 \end{theorem}
 \end{samepage}
 \noindent We note that our protocol only requires us to implement the qDRIFT channel multiple times using different time-step sizes and then take  measurements on the output states.
 Notably, it does not require any auxiliary qubits, additional controlled gates, or quantum operations beyond implementing the exponentials $e^{-iH_j t}$.
 The estimator can then be constructed by classical post-processing of the obtained measurement values.
 The protocol is thus highly parallelizable and requires minimal quantum resources.
 Additionally, the qFLO procedure is robust against noise in the measurements used to construct the estimator, although the overall error is limited by the measurement precision.

The qFLO algorithm inherits the advantages of the qDRIFT algorithm, as shown by the circuit depth comparison given for different simulation protocols in  \cref{Table:Circuit_Depth_Comparison}.
 In particular, the number of steps which are sufficient to reach a desired error is independent of the number of terms in the Hamiltonian decomposition $L$, where $H=\sum_{j=1}^L h_jH_j$.
 This is of particular interest for quantum chemistry applications, where Hamiltonians acting on $n$ fermionic modes may have $L=O(n^4)$ terms and hence introduces a factor of $L^2=O(n^8)$ which can be large.
 Even for modest sized quantum chemistry systems we typically have $n\sim 100$ leading to prohibitively high simulation costs \cite{poulin2014trotter}.
 We also realise that because the qFLO scheme involves extrapolating through multiple points using the qDRIFT channel, this implies that for a given maximum circuit depth, there is always a qFLO routine that performs at least as well as the qDRIFT routine.

Unfortunately, FLO also inherits the $O(\lambda^2T^2)$ scaling in the time-parameter from the qDRIFT protocol, meaning that higher-order Trotter-Suzuki formula will asymptotically out-compete it.
Furthermore, it does not obviously obtain the commutator scaling that is observed in Trotter-Suzuki formulae, and hence qFLO will likely be more competitive for Hamiltonians which have a less commuting structure such as those seen in molecular Hamiltonians.
 Finally we realise that a key limitation of qFLO relative to qDRIFT is that it only works for estimating observable expectation values to high precision, but does not prepare the time-evolved state to high precision.

\bgroup
\def\arraystretch{1.6} 
\begin{table}[h!]
    \centering
    \begin{tabular}{c|c}
    \multicolumn{2}{c}{\textbf{Circuit Depths for Different Simulation Methods}} \\ \hline \textbf{Method} & \textbf{Maximum Circuit Depth} \\
     \hline  \hline  \textbf{$1^{st}$-Order TS} \cite{childs2021theory} &  $L^3 (\alpha^{(1)} T)^2/\epsilon$ \\ 
         \textbf{$p^{th}$-Order TS} \cite{childs2021theory} & $L^{2+1/p} (\alpha^{(p)} T)^{1+1/p}/\epsilon^{1/p}$ \\
       \hline  \textbf{$p^{th}$-Order TS Random} \cite{childs2019faster} & $L^2 (\Lambda T)^{1+1/p}/\epsilon^{1/p}$ \\
    \hline  \textbf{  Extrap. $p^{th}$-Order TS } \cite{watson2024exponentially} & $ L^{2+1/p} (\beta^{(p)} T)^{1+1/p}\polylog(1/\epsilon)$ \\
    \hline   \textbf{qDRIFT} \cite{campbell2019random} &  $ (\lambda T)^2/\epsilon$ \\
      \textbf{  qFLO } [\textcolor{blue}{This work}] &$ (\lambda T)^2\log(1/\epsilon)$ \\
    \hline
    \end{tabular}
    \caption{A list of the asymptotic scaling of sufficient circuit depths to reach an error $\epsilon$ for measurement of an observable after time $T$ for a Hamiltonian $H = \sum_{i=1}^L h_jH_j$, $\lambda = \sum_j h_j,$ $\Lambda =\max_j  h_j$.
    ``TS'' denotes Trotter-Suzuki product formulae.
    Here $\alpha^{(p)}$, $\beta^{(p)}$ are a sets of $p$-times nested commutators given in Refs.~\cite{childs2021theory} and \cite{watson2024exponentially} respectively, where we have normalised them to remove the $L$ dependence.
    We see that if $\lambda =O(L^{1-\delta}\Lambda )$ for any $\delta>0$, then qFLO and qDRIFT can achieve better asymptotic scaling in the number of terms in the Hamiltonian.}
    \label{Table:Circuit_Depth_Comparison}
\end{table}
\egroup

\begin{figure}[!ht] \framebox{
		\begin{minipage}{.45\textwidth}
			\raggedright
   {\bf qDRIFT Protocol.} \\ \hrulefill \\ 
   {\bf Input:} A list of Hamiltonian terms $H=\sum_j h_j H_j $, a simulation time $T$, a classical oracle function SAMPLE() that returns an value $j$ from the probability distribution $p_j = h_j / (\sum_{j} h_j)$ a time-step size $t$, initial state $\rho_0$, and an oracle MEASURE$_A()$ which returns a measurement of $A$ on a state.  \\
			{\bf Output:} A number representing a quantum measurement.
		 \\ \hrulefill \\ 
   \begin{enumerate}
				\item $\lambda \leftarrow \sum_j h_j$ 
				\item $N\leftarrow \lceil T/t \rceil$
				\item $i \leftarrow 0$
                \item $\rho \leftarrow \rho_0$
				\item While $i < N$
				\begin{enumerate}
					\item $i \leftarrow i+1$
					\item $j \leftarrow$ SAMPLE() 
                    \item $\rho \leftarrow e^{i \lambda T H_j / N }\rho e^{-i \lambda t H_j / N }$
				\end{enumerate}	
                \item Return MEASURE$_A(\rho)$.
			\end{enumerate} 
	\end{minipage}  } \quad  \quad \ \ 
 \framebox{
		\begin{minipage}{.45\textwidth}
			\raggedright
            {\bf Richardson Estimator.} \\ \hrulefill \\
			{\bf Input:} All the inputs to the qDRIFT protocol, and a list of time-steps $T_\text{list}$.  \\
			{\bf Output:} An estimate to the time-evolved observable. \\ \hrulefill \\
   \begin{enumerate}
   \item $k\leftarrow 0$
   \item $N_\text{samples} \leftarrow \left\lceil \frac{\norm{\obsv}^2}{\epsilon^2} \log\left(\frac{2|T_{\text{list}}|}{\delta}\right)\right\rceil$
   \item for $t$ in $T_{list}$
   \begin{enumerate}
       \item $i \leftarrow 0$
       \item SUM $\leftarrow 0$
       
       \item While $i< N_\text{samples}$
       \begin{enumerate}
           \item $i \leftarrow i+1$
           \item SUM $\leftarrow$ SUM $+$ qDRIFT$(t)$
       \end{enumerate}
       \item $a_k \leftarrow \frac{1}{N_\text{samples}}$SUM
       \item $b_k \leftarrow \prod_{\substack{t_j\in T_{list} \\ t_k \neq t_j }} \frac{1}{1 - (t_k/t_j)}$
       \item $k\leftarrow k+1$
       \end{enumerate}
       \item ESTIMATOR $\leftarrow \sum_{j=1}^{|T_{list}|} b_j a_j$
       \item Return ESTIMATOR
   \end{enumerate}
	\end{minipage}  }
    \caption{\textit{Left:} Pseudocode for the qDRIFT protocol and measuring the final state. \textit{Right:} Pseudocode for computing the Richardson extrapolator with access to the qDRIFT protocol as a subroutine.}
    \label{Fig:qDRIFT}
    \end{figure}
    
    \begin{figure}
 \framebox{
 \begin{minipage}{1.0\textwidth}
			\raggedright
   {\bf qFLO Estimator.} \\ \hrulefill \\
 {\bf Input:} A list of Hamiltonian terms $H=\sum_j h_j H_j $, a simulation time $T$, initial state $\rho_0$, an observable $\obsv$, and a precision $\epsilon$, and a probability of success $\delta$.  \\
			{\bf Output:} An estimate to the time-evolved local observable. \\ \hrulefill \\
   \begin{enumerate}
				\item $\lambda \leftarrow \sum_j h_j$ 
				\item $m\leftarrow \lceil\log(1/\epsilon)  \rceil$  
                \item $R \leftarrow \sqrt{8}m/\pi$
				\item $i \leftarrow 0$
				\item While $i < m$
				\begin{enumerate}
					\item $t_i \leftarrow  \left(\frac{\log(m)}{\epsilon}\right)^{1/m}   \frac{1}{\lceil 4(8\lambda T)^2\rceil}\left\lceil \frac{R}{\sin\left(\frac{\pi(2m-1)}{8m}\right)} \right\rceil^2 \left\lceil  \frac{R}{\sin\left(\frac{\pi(2i-1)}{8m}\right)} \right\rceil^{-2}   $ 
                    \item $T_\text{list}$  \textbf{.}  append$(  t_i )$
                    \item $i \leftarrow i+1$
				\end{enumerate}	
                \item Return RICHARDSON$(\rho_0,T, \epsilon, H, \obsv)$.
			\end{enumerate} 
	\end{minipage}  }
 	\caption{  The overall qFLO procedure utilising the subroutines from \cref{Fig:qDRIFT}.}
	\label{Fig:qFLO}
 \end{figure}

 \subsection{Related Work}
 The qDRIFT algorithm has been extended in multiple ways \cite{ouyang2020compilation,hagan2023composite, kiss2023importance, cao2024marqsim, chakraborty2024implementing}.
 Here we briefly outline some of the additional related literature on quantum simulation using randomisation and post-processing methods.
 
 \paragraph{qSWIFT.} The qFLO protocol presented here should be compared to the recent qSWIFT algorithm for higher order error scaling of Trotter formulae with randomised compilation by \citeauthor{nakaji2023qswift} \cite{nakaji2023qswift}.
qSWIFT works by explicitly writing the qDRIFT channel as a function of the time-step, and then explicitly computing the higher-order terms using a quantum circuit with an ancillary qubit plus classical post-processing. 
 This requires the implementation of singly-controlled time-evolution operators of the Hamiltonian terms. 
 The $K^{th}$-order qSWIFT estimator achieves requires using circuit depths $O\left((\lambda T)^2/\epsilon^{1/K}\right)$ to measure time-evolved observables to relative error $\epsilon$.
 The post-processing is necessary as the higher order terms are not necessarily a CPTP map, and therefore cannot generally be physically implemented.
 We contrast this with the Richardson extrapolation approach of qFLO presented here which implicitly cancels-off the higher order terms rather than explicitly computing them.
 Both the qSWIFT and qFLO approaches suffer from the fact that they do not actually prepare the time-evolved state, but rather recover higher-precision measurements using post-processing.

\paragraph{Randomised Product Formulae.} Before qDRIFT was developed, randomised product formulae had been studied, where here the randomisation is done over how the summands are ordered.
Such product formulae have been shown to perform asymptotically better than standard scaling \cite{childs2019faster, cho2024doubling}.
We note that these algorithms are qualitatively different form the qDRIFT approach in that they randomly implement stages of product formulae, and each time-step requires implementing a number of terms proportional to the number in the Hamiltonian.

Recent work has also considered implementing terms probabilistically in such a way that the average gate count is independent of the precision $\epsilon$, although the maximum circuit depths still remain a function of precision \cite{granet2024hamiltonian}.

 \paragraph{Extrapolated Product Formulae. }
 Multiple previous works have used the idea of treating the time-step size a parameter which can be extrapolated over in a similar manner to this work.
 This was first done by \citeauthor{endo2019} \cite{endo2019} where Richardson extrapolation was used to empirically mitigate noise and algorithmic error.
 This was again used by \citeauthor{vazquez2023well}
 \cite{vazquez2023well} where Richardson extrapolation was again used to improve the measurement of observables.
 
 Separately \citeauthor{rendon2024improved} \cite{rendon2024improved} demonstrated that polynomial interpolation could be used to improve the error scaling in staged product formulae to $O(\log(1/\epsilon))$. 
 \citeauthor{watson2024exponentially} \cite{watson2024exponentially} and \citeauthor{rendon2024towards} \cite{rendon2024towards}  showed $p^{th}$-order staged product formulae retain the expected scaling $O(T^{1+1/p})$ time scaling when using polynomial interpolation, and the former further shows that commutator scaling still applies.

 Beyond the Richardson and polynomial extrapolation techniques above, other novel implementations have been developed such as parametric matrix models by \citeauthor{cook2024parametric} \cite{cook2024parametric}.
 However, as far as the author is aware, rigorous resource estimates and asymptotic scaling have not been given for this method.

 \paragraph{Quantum Simulation using Non-Unitary Channels.}
 Previous work has utilised non-unitary channels to implement approximations to quantum dynamics on a quantum computer.
 \citeauthor{faehrmann2022randomizing} \cite{faehrmann2022randomizing} uses a ``derandomised'' version of the well-conditioned linear combination of unitaries method presented in \citeauthor{low2019well} \cite{low2019well}. 
 A similar idea is used in the appendix of Ref.~\cite{nakaji2023qswift} where it is also noted that the randomised approach to phase estimation of \citeauthor{wan2022randomized} \cite{wan2022randomized} can be extended to time simulation of observables.
 Again, this result again requires ancillary qubits to run the simulation algorithm (and likely requires more gates than the qSWIFT protocol).
 \citeauthor{wang2024faster} \cite{wang2024faster} demonstrate improved performance by mixing algorithms which express the time-evolution using different series expansions.
  \citeauthor{martyn2024halving} \cite{martyn2024halving} uses a randomised compiling applied to quantum signal processing to improve the asymptotic scaling of quantum simulation, which improves the overall performance by a factor of a half.
 \citeauthor{zhuk2023trotter}
 \cite{zhuk2023trotter} show that post-processing can be used to give quadratically better scaling when using $p^{th}$-order Trotter-Suzuki formulae.
 \citeauthor{gong2023improved} \cite{gong2023improved} uses random orderings of Trotter-Suzuki operations plus post-processing to improve performance of simulation algorithms.

\section{Series Expansion for the Repeated qDRIFT Channel}
\label{Sec:Series_Expansion}

\subsection{Notation}

Let $X$ be a matrix acting on a finite dimensional Hilbert space $\mathcal{H}$.
We will use $\norm{X}_p$ to denote the Schatten $p$-norm of an operator, and $\norm{X}$ for the special case of the operator norm/Schatten $\infty$-norm.
For two matrices $X,Y$, we use $\ad_Y(X) =[Y,X] = YX-XY$. 
Let $\mathcal{T}:X\mapsto \mathcal{T}(X)$ be a linear map on the space of matrices of $\mathcal{H}$.
For a mapping $\cT$, we define the diamond distance as 
\begin{align}
    \dnorm{\cT}\coloneqq \sup_{\substack{\rho\in \mathcal{H}\otimes\mathcal{H} \\ \norm{\rho}_1\leq 1 } }\norm{(\cT \otimes \cI_{\mathcal{H}})(\rho)}_1.
\end{align}
We use $\cI_\mathcal{H}$ to denote the identity operator on a Hilbert space $\mathcal{H}$, and will drop the subscript when the dimension is clear from context.

Given a local Hamiltonian we will decompose it as $H=\sum_j h_j H_j$, where we chose the decomposition such that $h_j>0$ and $\norm{H_j}=1$ as per Ref.~\cite{campbell2019random}.
We denote $\lambda= \sum_j h_j$ and $\Lambda = \max_j h_j$.
$\lambda$ upper bounds the largest eigenvalue of $H$.
We assume that the exponentials $e^{-iH_jt}$ can be implemented without error for arbitrary times $t\in \mathbb{R}.$

\subsection{The Series Expansion}

If we wish to implement the time-evolution of a Hamiltonian $H=\sum_j h_j H_j$ using $N$ implementations of the qDRIFT channel, we set a timescale $t = T/N$, and define 
\begin{align}
    \cE(\rho) =  \sum_j p_j e^{-it \lambda\ad_{H_j} }(\rho)
\end{align}
where $\ad_B(A) = [B,A]$, \  $p_j = h_j/\lambda$.
We will find it convenient to work with the parameter $s\coloneqq 1/N$, and thus rewrite this as a parameterised set of channels
\begin{align}
    \cE_s(\rho) \coloneqq \sum_j p_j e^{-isT\lambda\ad_{H_j} }(\rho).
\end{align}
For mathematical convenience, we now want to represent $\cE_s$ in terms of a matrix exponential, which can be done provided the logarithm of $\cE_s$ exists. 
In \cref{Sec:Existence_of_Logarithm} we prove that it does provided $2sT < 1/\lambda $.
Assuming we are working in this regime, we can write:
\begin{align}
    \cE_s = e^{-G'} \quad \quad G' = -\log(\cE_s),
\end{align}
where $G'$ is defined by the matrix logarithm 
\begin{align}
    G' &=  - \sum_{k=1}^\infty \frac{(-1)^{k+1}}{k}(\cE_s-\cI)^k \\
       &= (\cI-\cE_s) + \sum_{k=2}^\infty \frac{1}{k}(\cI-\cE_s)^k. \label{Eq:Log_Series}
\end{align}
We prove this logarithm exists in \cref{Sec:Existence_of_Logarithm}.
We now consider the Taylor series of $\cE_s$:
\begin{align}
    \cE_s = \cI  - isT\ad_{H} + \delta(s)
\end{align}
where $\delta(s)$ has a lowest order of $O((sT)^2)$.
This can be rewritten as
\begin{align}
    \cI-\cE_s = isT\ad_{H} +\delta(s).
\end{align}
We then realise that if we substitute this expansion into the expression for $G'$, then the first term in the expansion of $G'$ is $O(sT)$.
Thus we can rewrite $G'$ as
\begin{align}
    \cG(s) \coloneqq \frac{1}{isT} G' \quad  \text{and hence}  \quad \cE_s = e^{-isT\cG(s)},
\end{align}
where $\cG(s)$ can be thought of as the generator for the mapping.
We will also find it useful to write $G'$ and $\cG$ in terms of a series expansion
\begin{align}
    G &= i(sT\ad_H + \sum_{j=2}^\infty E_j (sT)^j) \\
    \implies \cG(s) &= \ad_H + \sum_{j=1}^\infty E_{j+1} (sT)^j, \label{Eq:cG_Power_Series}
\end{align}
where $\{E_j\}_{j=1}^\infty$ are super-operator coefficients defined from \cref{Eq:Log_Series}. 

\subsection{Repeated Applications of the Mapping}
We can write the channel in terms of the series expansion for $\cG(s)$ as
\begin{align}
    \cE_{s} = e^{-i(sT\ad_H + \sum_{j=2}^\infty E_j (sT)^j)}.
\end{align}
We now consider $1/s$-many repeated applications of the mapping.
\begin{align}
    \cE_{s}^{1/s} &= e^{-iT\cG(s)} \\
    &= e^{-iT(\ad_H + \sum_{j=1}^\infty E_{j+1} (sT)^j)} \\
    &= e^{-iT(\ad_H+\Delta(s))},
\end{align}
where we can write
\begin{align}
    \cG(s) &= \ad_H+ \Delta(s) \\
    &= \ad_H + \sum_{j=1}^\infty E_{j+1}s^jT^{j}, \label{Eq:E_j_Definition}
\end{align}
where $\Delta(s) = \sum_{j=1}E_{j+1}s^jT^j$
and where we have absorbed a factor of $\lambda^j$ into $E_j$.
We show later in \cref{Lemma:Ek_Bound} that $\dnorm{E_k}\leq (4\lambda )^k$.
We now want to find an expansion of $e^{T\cG(s)}$ in terms of the inverse step-size $s$. 
To do this, we realise that if we treat $s$ as a constant (i.e. we fix a Trotter-step size), and define $\tilde{\state}(s,T) = e^{-iT\cG(s)}\state$ to be the state time-evolved under the qDRIFT mapping, then the evolution of the state satisfies the following differential equation
\begin{align}
    \partial_T\tilde{\state}(s,T) = -i\cG(s)\tilde{\state}(s,T).
\end{align}
\begin{lemma}[Variation-of-Parameters Formula]
\label{Lemma:Variation_of_Parameters}
    Given an ODE of the form $\partial_x O  = (A+B)O,$ we can write the solution as
    \begin{align}
        e^{A+B}O = e^AO + \int_0^1 dx \ e^{(1-x)A} B e^{x(A+B)}O.
    \end{align}
\end{lemma}
\noindent We refer to \cite[Lemma 6]{aftab2024multi} for an explicit proof of the lemma in this particular form.
This can then be used to give an expansion of observables in terms the inverse step size, $s$, in the following way.
\begin{lemma} \label{Lemma:Series_Expansion}
   Defining $\tstate(s,T) = e^{-iT\cG(s)}\state,$ to be the approximate time-evolved observable and $\state(T) = e^{-iT\ad_H}\state$ to be the Hamiltonian evolution.
   Then we can write this as a series expansion
   \begin{align}
       \tilde{\state}(s,T) - \state(T) = \sum_{j\geq2} s^j \tilde{E}_{j+1,K}(T)(\state) + \tilde{F}_K(T,s)(\state),
   \end{align}
   for any $K>3$.
   We have the bounds
   \begin{align}
        \dnorm{\tilde{E}_{j+1,K}(T)} &\leq (8\lambda T)^j \sum_{l=1}^{\min\{K-1, j\}} \frac{(8\lambda T)^l}{l!} \\
       \dnorm{\tilde{F}_K(T,s)}&\leq  \frac{(8\lambda T)^K}{K!}\sum_{j=1}^\infty(8\lambda sT)^{j}. 
    \end{align}
\end{lemma}
\begin{proof}

We can apply the variation-of-parameters formula from \cref{Lemma:Variation_of_Parameters} to expand the repeated qDRIFT channel as a series in $s$:
\begin{align}
    e^{-iT(\ad_H+\Delta(s))} = e^{-iT\ad_H} + \int_0^1 dx_1 e^{-iT(1-x_1)\ad_H} (-iT)\Delta(s)e^{-iT(\ad_H+\Delta(x_1))},
\end{align}
where we have set $A = -iT\ad_H$ and $B=-iT\Delta(s)$.
We can then iterate this with $K$-many repeated substitutions of the variation-of-parameters formula
\begin{align}
    e^{-iT(\ad_H+\Delta(s))}(\state) &= e^{-iT\ad_H}(\state) \\ 
        &+ \sum_{l=1}^{K-1} (-T)^l \int^1_0 dx_1 \int^{x_1}_0 dx_2 \dots \int^{x_{l-1}}_0 dx_l e^{iT(1-x_1)\ad_H } i\Delta(s) e^{i(x_1-x_2)T\ad_H } i\Delta(s) \dots  e^{i(x_{l-1}-x_l)T\ad_H } i\Delta(s)(\state(x_lT))   \label{Eq:Iterate_K}
         \\
         &+ (-T)^K\int^1_0 dx_1\int^{x_1}_0 dx_2 \dots \int_0^{x_{p-1}} dx_K e^{i(1-x_1)T\ad_H } i\Delta(s) e^{i(x_1-x_2)T\ad_H }i\Delta(s) \dots \\ 
         &\quad \dots  \times    e^{i(x_{K-1}-x_K)T\ad_H } i\Delta(s)(\tilde{\state}(s,x_KT)). \label{Eq:Remainder_Term}
\end{align}
For the following we denote $t=sT$. 
We will consider \cref{Eq:Iterate_K} and \cref{Eq:Remainder_Term} separately.
First examine \cref{Eq:Iterate_K} and expand it using the definition of $E_j$ from \cref{Eq:E_j_Definition} to get  
    \begin{align}
        &\sum_{l=1}^{K-1} (-T)^l \int^1_0 dx_1\int^{x_1}_0 dx_2\dots \int^{x_{l-1}}_0 dx_l \  e^{i(1-x_1)T\ad_H } i\Delta(s) e^{i(x_1-x_2)T\ad_H} i\Delta(s) \dots  e^{ix_{l-1}-x_l)T\ad_H } i\Delta(s)(\state(x_lT)). 
        \nonumber\\ 
        =& \sum_{l=1}^{K-1} (-T)^l \int^1_0 dx_1 \int^{x_1}_0 dx_2  \dots \int^{x_{l-1}}_0 dx_l \left( \prod_{\kappa=l}^1\left( \sum_{j_\kappa\in \mathbb{Z}_+\geq 2 } e^{i(x_{\kappa-1}-x_\kappa)T\ad_H}    iE_{j_\kappa+1} t^{j_\kappa}  \right)   \right)(\state(x_lT)) 
        \\ 
        =&\sum_{l=1}^{K-1} (-T)^l \int^1_0 dx_1 \int^{x_1}_0 dx_2 \dots \int^{x_{l-1}}_0 dx_l \left( \sum_{\substack{j\in \mathbb{Z}_+\geq l}} t^j\sum_{\substack{j_1\dots j_l\in \mathbb{Z}_+\geq 2\\ j_1+\dots+ j_l=j}}\left(\prod_{\kappa=l}^1 e^{i(x_{\kappa-1}-x_\kappa)T\ad_H}    iE_{j_\kappa+1}\right) \right)(\state(x_lT)) \label{Eq:E_Derivation}
    \end{align}
    We can then regroup the terms in terms of $s$:
    \begin{align}
        &\sum_{j\geq 2}(sT)^j\sum_{l=1}^{\min\{K-1, j\}} (-T)^l \int^1_0dx_1\int^{x_1}_0dx_2\dots \int^{x_{l-1}}_0dx_{l-1}\sum_{\substack{j_1,\dots,j_l\in \mathbb{Z}\geq 1 \\ j_1+\dots + j_l=j}}\left(  \prod_{\kappa=l}^1 e^{iT(x_{\kappa-1}-x_\kappa)\ad_H} \ iE_{j_\kappa+1}   \right)(\state(s_lT)) \\
        &= \sum_{j\geq 2} s^j \tilde{E}_{j+1,K}(T)(\state).
    \end{align}
    Here we have defined
    \begin{align}
        \tilde{E}_{j+1,K}(T)(\rho) = \sum_{l=1}^{\min\{K-1, j\}} (-T)^{j+l} \int^1_0dx_1\int^{x_1}_0dx_2\dots \int^{x_{l-1}}_0dx_{l-1}\sum_{\substack{j_1,\dots,j_l\in \mathbb{Z}\geq 1 \\ j_1+\dots + j_l=j}}\left(  \prod_{\kappa=l}^1 e^{iT(x_{\kappa-1}-x_\kappa)\ad_H} \ iE_{j_\kappa+1}   \right)(\state(x_lT))
    \end{align}
    If we now want to put bounds on this quantity, we can apply the triangle inequality and submultiplicativity of the norm to show
    \begin{align}
        \dnorm{\tilde{E}_{j+1,K}(T)}\leq T^j \sum_{l=1}^{\min\{K-1, j\}} T^l   \int^1_0dx_1\int^{x_1}_0dx_2\dots \int^{x_{l-1}}_0dx_{l-1} \sum_{\substack{j_1,\dots,j_l\in \mathbb{Z}\geq 1 \\ j_1+\dots + j_l=j}}\left(  \prod_{\kappa=l}^1  \dnorm{E_{j_\kappa+1}}   \right)
    \end{align}
    Then using that $\dnorm{E_k}\leq (4\lambda)^k$ from \cref{Lemma:Ek_Bound} (proven in the appendix), we can rewrite this as 
    \begin{align}
        \dnorm{\tilde{E}_{j+1,K}(T)} &\leq T^j \sum_{l=1}^{\min\{K-1, j\}} T^l \int^1_0dx_1\int^{x_1}_0dx_2\dots \int^{x_{l-1}}_0dx_{l-1}\sum_{\substack{j_1,\dots,j_l\in \mathbb{Z}\geq 1 \\ j_1+\dots + j_l=j}}\left(  \prod_{\kappa=l}^1  (4\lambda )^{(j_\kappa+1)}   \right) \\
        &\leq T^j \sum_{l=1}^{\min\{K-1, j\}} \frac{T^l}{l!} \sum_{\substack{j_1,\dots,j_l\in \mathbb{Z}\geq 1 \\ j_1+\dots + j_l=j}}\left(  \prod_{\kappa=l}^1  (4\lambda )^{(j_\kappa+1)}   \right) \\
        &\leq T^j \sum_{l=1}^{\min\{K-1, j\}} \frac{(4T\lambda)^l}{l!} \sum_{\substack{j_1,\dots,j_l\in \mathbb{Z}\geq 1 \\ j_1+\dots + j_l=j}}\left(  \prod_{\kappa=l}^1  (4\lambda )^{j_\kappa}   \right) \\
        &\leq (4\lambda T)^j \sum_{l=1}^{\min\{K-1, j\}} \frac{(4T\lambda)^l}{l!} \sum_{\substack{j_1,\dots,j_l\in \mathbb{Z}\geq 1 \\ j_1+\dots + j_l=j}} 1  
    \end{align}
    We now want to bound the following
    \begin{align}\label{Eq:Combinatorial_Identity}
       \sum_{\substack{j_1,\dots,j_l\in \mathbb{Z}\geq 1 \\ j_1+\dots + j_l=j}}   1     = \binom{j+l-1}{l-1} < 2^{j+l}
    \end{align}
    Substituting this back in gives
    \begin{align}
        \dnorm{\tilde{E}_{j+1,K}(T)} \leq (8\lambda T)^j \sum_{l=1}^{\min\{K-1, j\}} \frac{(8T\lambda)^l}{l!}. 
    \end{align}

    We now want to put bounds on the term \cref{Eq:Remainder_Term} which we label 
    \begin{align}
       \tilde{F}_{K} =  (-T)^K\int^1_0 dx_1\int^{x_1}_0 dx_2 \dots \int_0^{x_{p-1}} dx_K e^{i(1-x_1)T\ad_H } i\Delta(s) e^{i(x_1-x_2)T\ad_H } i\Delta(s) \dots  e^{i(x_{K-1}-x_K)T\ad_H } i\Delta(s)(\tilde{\state}(s,x_KT)).
    \end{align}
    To norm of  $\tilde{F}_{K}$, we see that 
    \begin{align}
        \dnorm{\tilde{F}_K(s,T) }&\leq \frac{T^K}{K!}\dnorm{\Delta(s)} \\
        &\leq  \frac{T^K}{K!}\left( \sum_{j=1}\dnorm{E_{j+1}}(sT)^j \right)^K\\
        &\leq \frac{T^K}{K!} \sum_{j_1\dots j_K=1} \left(\prod_{\kappa=1}^K\dnorm{E_{j_k+1}}\right)(sT)^{j_1+\dots +j_K} \\
        &\leq \frac{(4\lambda T)^K}{K!}\sum_{j=1}(sT)^{j} \sum_{j_1+\dots +j_\kappa=j}(4\lambda )^{j_k} \\
         &\leq \frac{(4\lambda T)^K}{K!}\sum_{j=1}(4\lambda sT)^{j} \sum_{j_1+\dots +j_\kappa=j} 1 \\
         &\leq \frac{(8\lambda T)^K}{K!}\sum_{j=1}(8\lambda sT)^{j}
    \end{align}
    where in the last line we have used \cref{Eq:Combinatorial_Identity}.
\end{proof}

\section{Robust Richardson Extrapolation}
\label{Sec:Richarson_Extrapolation}

We now wish to use a Richardson extrapolation scheme to construct an estimator for the expectation value of a time-evolved observable.
To do this, we can write the expectation value as a power series using \cref{Lemma:Series_Expansion} to give
\begin{align}\label{Eq:Function_Series}
    f_A(s)  = \expval{A(T)} + \sum_{j=2}^\infty \alpha_j s^j + h_K(s)
\end{align}
where $h_K(s)$ is a function that only has terms of order $O(s^K)$ and above.
If we wish to do a Richardson extrapolation procedure, we need to choose a set of points which we evaluate to construct our estimator.
We choose the points $\{s/y_j\}_{j=1}^m$, which when plugged into the series gives
\begin{align}
     f_A\left(\frac{s}{y_i}\right)  = \expval{A(T)} + \sum_{j=2}^\infty \alpha_j \left(\frac{s}{y_i}\right)^j + h_K\left(\frac{s}{y_i}\right)
\end{align}
Suppose we now wish to construct an estimator
\begin{align}\label{Eq:Estimator_Form}
    \estimator_m = \sum_{j=1}^m b_j f_A\left(\frac{s}{y_j}\right),
\end{align}
and we wish to choose values of $y_j$ such that the estimator is robust to errors in the calculation of $f_A$.
In particular, for an arbitrary choice of values $\{y_j\}_{j=1}^m$, small errors in the calculated value $f_A(y_j)$ can ``blow up'' to ruin the final estimator. 
We see that for a generic power series, the solution to find the coefficients $b_j$ such that all terms up to $O(s^m)$ are cancelled off, is given by the Vandermonde matrix equation
\begin{align}
    \begin{pmatrix}
        1 & 1 & \dots & 1 \\
        y_1^{-1} & y_2^{-1} & \dots & y_m^{-1} \\
        \vdots & \vdots &\ddots & \vdots \\
        y_1^{-m+1} & y_2^{-m+1} & \dots & y_m^{-m+1}
    \end{pmatrix}
    \begin{pmatrix}
        b_1 \\
        b_2 \\
        \vdots \\
        b_m
    \end{pmatrix}
    =
    \begin{pmatrix}
        1 \\
        0 \\
        \vdots \\
        0
    \end{pmatrix}.
\end{align}
For convenience, we have not automatically set the coefficient $\alpha_1=0$ as is implied by the form of the series in \cref{Eq:Function_Series}.
If we choose a relabelling $y_j=k_j^2,$ then the problem reduces to that studied in Ref.~\cite{low2019well} which will allow us to simply import the solution studied there.
In particular, we see that the $b_j$ coefficients are given by
\begin{align}\label{Eq:b_i_definition}
    b_j = \prod_{l\neq j} \frac{1}{1 - (k_l/k_j)^2}.
\end{align}
We straightforwardly take the solution for this from Ref.~\cite{low2019well} and choose
\begin{align}\label{Eq:Sampling_Points}
    y_j = k_j^2 = \left\lceil \frac{R}{\sqrt{x^{(2m)}_j}}\right\rceil^2 \quad \quad j\in \{1,2,\dots, m\}
\end{align}
where we choose to set $R = \frac{\sqrt{8}m}{\pi}$ and
\begin{align}
    x^{(k)}_j = \sin^2\left( \frac{\pi(2j-1)}{4k}\right).
\end{align}

\begin{lemma}[\cite{low2019well}] \label{Lemma:Step_Scaline_m}
    For the choice of $\{y_j\}$ given in \cref{Eq:Sampling_Points} gives $\norm{b_1}=O(\log(m))$, $\max_i y_i = O(m^4)$ and $\max_jy_j / y_m=O(m^2)$.
\end{lemma}
\begin{proof}
    The result for $\norm{b}_1$ follows directly from Ref.~\cite{low2019well}.
    The result for $\max_iy_i$ follows from noting that $R/\sqrt{x^{(2m)}_1}=O(m^2)$.
    Finally we see that $\max_j(y_j/y_m) = O\left( R^2/x^{(2m)}_1\times x^{(2m)}_m/R^2\right) = O(m^2)$.
\end{proof}

We can now consider constructing a well-conditioned Richardson extrapolator for time-evolved observables under the qDRIFT channel.
We see that we can construct a well-conditioned estimator of the form in \cref{Eq:Estimator_Form} by using the step-sizes given in \cref{Lemma:Step_Scaline_m}. 
Thus we define the $m^{th}$-order Richardson estimator as follows
\begin{align}
    \estimator_m(s,T) \coloneqq \sum_{j=1}^m b_j\expval{\tobsv(s_j,T)}
\end{align}
where we define $s$ 
\begin{align}
    s\equiv s_m = \frac{\ell}{y_m}, \quad \quad  s_j = \frac{\ell}{y_j} = s\frac{y_m}{y_j}  \quad \quad j\in \{1,2,\dots m\},
\end{align}
where $\ell$ is some parameter we will fix later by choosing an appropriate $s_1$.
Since rescaling by a factor $sy_1$ does not change the values of $b_j$ (from \cref{Eq:b_i_definition} we see that $b_j$ is defined in terms of ratios), then $\norm{b_1}=O(\log(m))$.
Importantly for our purposes, we see that $\max_i \frac{1}{s_i} = \max_i N_i = y_j/(s_m y_m) = O(m^2/s_1)$ which follows from \cref{Lemma:Step_Scaline_m}.

\section{Error Bounds from Richardson Extrapolation}
\label{Sec:Richardson_Bounds}

Here we show that there exists a Richardson estimator which can be computed using only $m=O(\log(1/\epsilon))$ many sample points, where $\epsilon$ is the desired precision we wish to know the time-evolved observable expectation to.

\begin{lemma} \label{Lemma:Richardson_Error}
    Let $\estimator_{m}(s, T)$ be the $m^{th}$-order Richardson estimator defined as 
    \begin{align*}
        \estimator_{m}(s, T) =\sum_{k=1}^m b_k \expval{\tilde{\obsv}(s_k,T)}, 
    \end{align*}
    where $N_k \in \mathbb{Z}$ and $s_k=1/N_k$. 
    Then the error in the extrapolation relative to the exact evolution is given by
    \begin{align}
        | \expval{\obsv(T)} - \estimator_{m}(s, T)   |\leq \norm{\obsv} \norm{b}_1  \sum_{j\geq m} s^j (8\lambda T)^j \sum_{l=1}^{K} \frac{(8\lambda T)^l}{l!},
    \end{align}
    where we assume $ K>m-1$.
\end{lemma}
\begin{proof}
    First consider the series expansion from \cref{Lemma:Series_Expansion}.
    The $m^{th}$-order Richardson estimator removes all the terms up to $O(s^{m-1})$ inclusive.
    If we choose $K \geq  m$ then we get an estimator that satisfies
    \begin{align}
        \estimator_{m}(s, T) = \expval{\obsv(T)}+ \sum_{k=1}^m b_k R_m[\expval{\tobsv(s_k,T)}]
    \end{align}
    where $R_m[\expval{\tobsv(s_k,T)}]$ is the remainder of the series for all terms of $s^m$ and above
    \begin{align}
        R_m[\expval{\tobsv(s,T)}] = \sum_{j\geq m} s^j \tr\left[\obsv \ \tilde{E}_{j+1,K}(T)(\state)\right] + \tr\left[\obsv \  F_{K}(T,s)(\state)\right].
    \end{align}
    We can then bound the error of the Richardson estimator in the following way using the submultiplicativity of the norm
    \begin{align}\label{Eq:Expectation_Value_Difference}
        |\estimator_{m}(s, T)-\expval{\obsv(T)}|\leq \norm{b}_1 \max_k \dnorm{R_m[\expval{\tobsv(s_k,T)}]}.
    \end{align}
    We see this maximum occurs for $k=m$.
    We can then bound the remainder term using \cref{Lemma:Series_Expansion} and H\"older's inequality:
    \begin{align}
        \max_k\dnorm{R_m[\expval{\tobsv(s_k,T)}]}& \leq \sum_{j\geq m} s_m^j \norm{\obsv} \dnorm{\tilde{E}_{j+1,K}(T)(\state)} + \norm{\obsv}\dnorm{\tilde{F}_K(T,s_m)(\state)} \label{Eq:Remainder_Bound} \\
        &\leq \norm{\obsv}  \left( \sum_{j\geq m} s_m^j (8\lambda T)^j \sum_{l=1}^{K-1} \frac{(8\lambda T)^l}{l!} +  \frac{(8\lambda T)^K}{K!}\sum_{j\geq m}^\infty(8\lambda s_mT)^{j}\right) \\
        &= \norm{\obsv} \sum_{j\geq m} s_m^j (8\lambda T)^j \sum_{l=1}^{K} \frac{(8\lambda T)^l}{l!},
    \end{align}
    where in \cref{Eq:Remainder_Bound} we have used both the triangle inequality and the submultiplicativity of matrix norms.
    Combining this with \cref{Eq:Expectation_Value_Difference} gives the lemma statement.
\end{proof}

\noindent We now wish to bound the number of Trotter steps needs to reach a given precision in the Richardson estimator.
\begin{lemma}\label{Lemma:Sufficient_Time_Steps}
    It is possible to construct an order-$m$ Richardson estimator, $\estimator_{m}(s, T)$, satisfying
    \begin{align}
         |\expval{A(T)} - \estimator_{m}(s, T)|\leq \epsilon \norm{\obsv}, 
    \end{align}
    where $s=s_m$ and
    \begin{align}
       \frac{1}{s_m} =  N_m = 4(8\lambda T)^2 \left(\frac{\norm{b}_1}{\epsilon}\right)^{1/m}.
    \end{align}
\end{lemma}
\begin{proof}
    We can take the error expression from \cref{Lemma:Richardson_Error}.
    We first denote $\eta = \max\{  8\lambda T, 1 \}$.
    We see that 
    \begin{align}
        \sum_{l=1}^K \frac{(8\lambda T)^l}{l!} &\leq \eta^K \sum_{l=1}^K \frac{1}{l!} \\
        &\leq \eta^K (e-1) \\
        &\leq 2 \eta^K.
    \end{align}
    Then the full error of the estimator is bounded as
    \begin{align}
        & \leq 2\norm{\obsv}\norm{b_1} \eta^K \sum_{j\geq m} (8\lambda s T)^j \\
        &\leq 2\norm{\obsv}\norm{b_1} \eta^K(8\lambda sT)^m \sum_{j=0}^\infty (8\lambda sT)^j
    \end{align}
    We have free choice over $K$ provided $K>m-1$.
    Setting $K=m$, we have
    \begin{align}
        &\leq 2\norm{\obsv}\norm{b_1} ((8\lambda T)^2s)^m \sum_{j=0} (8\lambda sT)^j \\
        &\leq 2\norm{\obsv}\norm{b_1} ((8\lambda T)^2s)^m \frac{1}{1- (8\lambda T s)},
    \end{align}
    where we have assumed the geometric series converges.
    To ensure this becomes arbitrarily small as $m$ increases, we set $1/s = 4(8\lambda T)^2 \left(\frac{\norm{b}_1}{\epsilon}\right)^{1/m}$.
    With this choice we see that for sufficiently large $\lambda T$ then  $\frac{1}{1- (8\lambda T s)}<2$, and hence we have that 
    \begin{align}
         |\expval{A(T)} - \estimator_{m}(s, T)|\leq \epsilon \norm{\obsv}, 
    \end{align}
    as per the lemma statement.
\end{proof}

\noindent These bounds tell us the number of repeated steps that are sufficient to reach a certain error. 
We now show that with an appropriate choice of $m,$ the maximum circuit depth of the circuits scales as  $O(\log(1/\epsilon))$ asymptotically.
\begin{lemma}\label{Lemma:Step_Number}
    Let $\estimator_{m}(s, T)$ be an order-$m$ Richardson estimator.
    To reach precision $\epsilon,$ it is sufficient to use $m=O\left(\frac{\log(1/\epsilon)}{\log\log(1/\epsilon)}\right)$ sample points, with a maximum number of time-steps scaling as 
    \begin{align}
        \max_i N_i = O\left( (\lambda T)^2 \frac{\log(1/\epsilon)}{(\log\log(1/\epsilon))^2} \right)
    \end{align}
\end{lemma}
\begin{proof}
    We see that a sufficient number of time-steps to achieve an error $\epsilon \norm{\obsv} $ is given in \cref{Lemma:Sufficient_Time_Steps}.
    Noting that $\norm{b}_1 = O(\log(m))$ then $\norm{b}_1^{1/m} = O(\log^{1/m}(m)) = O(1).$
    We can then choose 
    \begin{align}
        \frac{1}{\epsilon^{1/m}} = \left( \frac{m}{\log^2(1/\epsilon)}\right)
    \end{align}
    Substituting this into the expression for $N_m$ from \cref{Lemma:Sufficient_Time_Steps}, we have
    \begin{align*}
        N_{\max} = O\left( (\lambda T)^2 \frac{m^3}{\log^2(1/\epsilon)} \right).
    \end{align*}
    We now define $x = \frac{m}{\log^2(1/\epsilon)}$, we can write that
    \begin{align}
        \frac{1}{\epsilon} &= x^{x\log^2(1/\epsilon)} \\
        \implies x^x &= \left( \frac{1}{\epsilon} \right)^{1/\log^2(1/\epsilon)} \\
        &=: a.
    \end{align}
    Solving for $x$ then gives
    \begin{align}
        x = O\left( \frac{\log(a)}{\log\log(a)} \right)
    \end{align}
    which gives us
    \begin{align}
        m = x\log^2\left(\frac{1}{\epsilon}\right)=O\left( \frac{\log(1/\epsilon)}{\log\log(1/\epsilon)}\right)
    \end{align}
    which gives a maximum number of steps as
    \begin{align}
        N_{\max} = O\left( (\lambda T)^2 \frac{m^3}{\log^2(1/\epsilon)} \right) = O\left( (\lambda T)^2 \frac{\log(1/\epsilon)}{(\log\log(1/\epsilon))^3} \right).
    \end{align}
    Finally, the total number of gates used goes as $O(mN_{\max})$.
    This gives
    \begin{align}
        mN_{\max} =  O\left( (\lambda T)^2 \frac{\log^2(1/\epsilon)}{(\log\log(1/\epsilon))^2} \right).
    \end{align}

\end{proof}

\section{Robustness to Measurement Error}
\label{Sec:Robust_Richardson}

So far we have assumed that when constructing the Richardson estimator, we have perfect measurements. 
Naturally, if we are taking data in practice, this will not be the case, and there will be some measurement error associated with the system.
We need to ensure that our estimator is well-conditioned, and will not be too adversely affected by these errors.
Let $\estimator_m$ be the Richardson estimator computed using perfect data, and let $\testimator_m$ be the Richardson estimator computing with data where each point may have error up to $\epsdata$.
\begin{align}
    |f_A(0) - \testimator_m(s)| &\leq |f_A(0) - \estimator_m(s,T)| + |\estimator_m(s) - \testimator_m(s,T)|  \nonumber \\
     &\leq |f_A(0) - \estimator(s,T)| + \left|\sum_{k=0}^{m-1} b_k f_A(s_k) - \sum_{k=0}^{m-1} b_k \tilde{f}(s_k) \right|  \nonumber \\
    &\leq \norm{\obsv}(\epsint + \norm{b}_1\epsdata). \label{Eq:Error_Decomposition}
\end{align}
Thus, if we want to each an overall error of $\epsilon$ in our final estimate, we can make the following arbitrary partition of the error
\begin{align}
    \epsint = \frac{\epsilon}{2}, \quad \quad  \epsdata = \frac{\epsilon}{2\norm{b}_1}.
\end{align}
Since $\norm{b}_1 = O(\log(m))=O(\log\log(1/\epsilon))$ as per \cref{Lemma:Step_Scaline_m}, this does not require significant overheads.
From here we realise there are two potential methods for estimating the data points at a particular point. 
We can apply a straightforward measurement of the observable $A$. 
Generically, this requires a number of samples that scales as $O(1/\epsdata^2)=\tilde{O}(1/\epsilon^2)$, but the only requirement is that we run the qDRIFT channel for the appropriate number of iterations.

 \begin{table}[H]
        \centering
        \begin{tabular}{c||c|c}
             & \textbf{Max Circuit Depth} & \textbf{Total Gate Cost} \\
           \hline\hline
           \textbf{Resource Scaling}    & $O\left( (\lambda T)^{2}\log\big( \frac{1}{\epsilon}\big) \right)$  & $O\left( \frac{(\lambda T)^{2}}{\epsilon^2}\log^2\big( \frac{1}{\epsilon}\big) \right)$
        \end{tabular}
        \label{Table:Resource_Costs}
    \end{table}
    \noindent This proves the main result given in \cref{Theorem:Main_Theorem}.

\section{Numerics} 
Here we test the above results using the Heisenberg model on a length $L$ 1D chain, defined as:
\begin{align*}
    H = \sum_{i=1}^{L-1} \left( X_{i}X_{i+1} + Y_{i}Y_{i+1} +Z_{i}Z_{i+1} \right) + \sum_{i=1}^{L} \mu_i Z_i
\end{align*}
where we choose the values $\{\mu_i\}_{i=1}^L$ uniformly randomly in the interval $[-1,1]$.
This Hamiltonian has been well studied in the quantum simulation literature such as in Refs. \cite{childs2018toward, watson2024exponentially}.

\begin{figure}[h!]
    \centering
    \begin{minipage}{0.45\textwidth}
        \centering
        \includegraphics[width=1.0\textwidth]{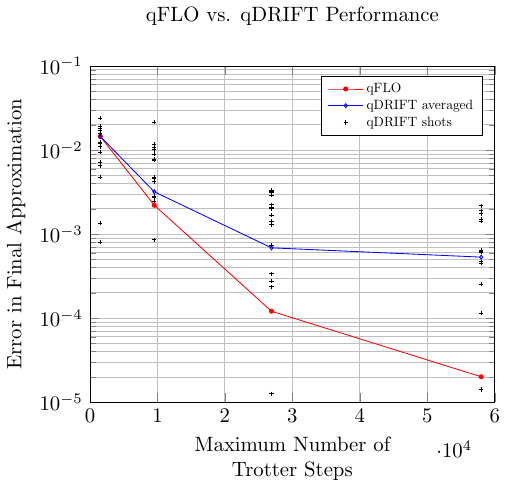}
    \end{minipage}\hfill
    \begin{minipage}{0.45\textwidth}
        \centering
        \includegraphics[width=1.0\textwidth]{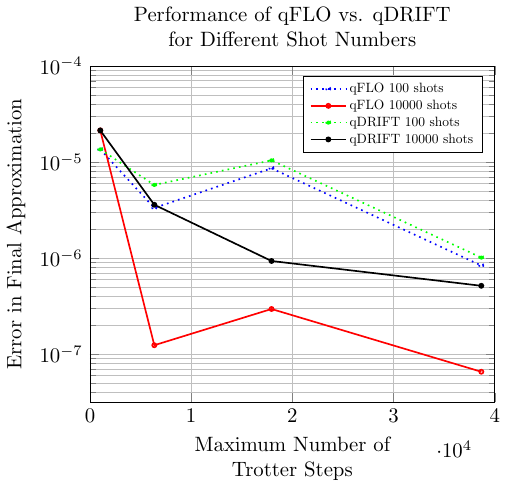} 
    \end{minipage}
    \caption{Left: qDRIFT and qFLO comparison for an $L=4$ Heisenberg model. 
    Each `+' represents 100 shots of qDRIFT averaged together, and the blue line represents an average value.
    Right: qDRIFT and qFLO comparison, but with different numbers of shots. 
    The dashed lines show where both qFLO and qDRIFT are limited by shot precision, whereas the solid lines show where only qFLO is limited by the shot precision.
    }
    \label{Fig:Fixed_Time_Change_Degree}
\end{figure}

We demonstrate the performance to qFLO vs. qDRIFT in \Cref{Fig:Fixed_Time_Change_Degree} (left).
Unlike the extrapolation routine for deterministic product formulae, here we need to average over a much larger number of shots to get comparable accuracy.
\Cref{Fig:Fixed_Time_Change_Degree} (right) demonstrates how the shot number can limit performance.

 \section*{Acknowledgements}
 {\begingroup
		\hypersetup{urlcolor=navyblue}
  The author would like to thank \href{https://orcid.org/0000-0001-8564-446X}{Mohsen Bagherimehrab} for his extensive feedback on this work, and in particular for demonstrating that the scaling in \cref{Lemma:Step_Number} can be improved from $O(\log^2(1/\epsilon))$ to  $O(\log(1/\epsilon))$. 
  The author would like to acknowledge very useful discussions with \href{https://orcid.org/0000-0003-1478-7230}{Jacob Watkins}, Mauro Morales, \href{https://orcid.org/0000-0003-2848-1216}{Yuxin Wang}, \href{https://orcid.org/0000-0003-0105-7677}{Christopher Kang} and \href{https://orcid.org/0000-0002-0335-9508}{Victor Albert}.
  He would further like to thank \href{https://orcid.org/0000-0002-3501-5734}{Kohei Nakaji} for useful feedback on the paper.
   He also acknowledges the support of his fianc\'ee Florence, after whom the qFLO algorithm presented in this work is named.
  		\endgroup}
 
 J.D.W. acknowledges support from the United States Department of Energy, Office of Science, Office of Advanced Scientific Computing Research, Accelerated Research in Quantum Computing program, and also NSF QLCI grant OMA-2120757.

{\begingroup
		\hypersetup{urlcolor=navyblue}
\printbibliography[heading=bibintoc]

@inbook{sidi_2003, 
    place={Cambridge}, 
    series={Cambridge Monographs on Applied and Computational Mathematics}, 
    title={The Richardson Extrapolation Process}, DOI={10.1017/CBO9780511546815.003}, 
    booktitle={Practical Extrapolation Methods: Theory and Applications}, 
    publisher={Cambridge University Press}, 
    author={Sidi, Avram}, 
    year={2003}, 
    pages={21–41}, 
    collection={Cambridge Monographs on Applied and Computational Mathematics},
    doi={https://doi.org/10.1017/CBO9780511546815.003}
}

@article{rendon2024improved,
  title={Improved Accuracy for Trotter Simulations Using Chebyshev Interpolation},
  author={Rendon, Gumaro and Watkins, Jacob and Wiebe, Nathan},
  journal={Quantum},
  volume={8},
  pages={1266},
  year={2024},
  publisher={Verein zur F{\"o}rderung des Open Access Publizierens in den Quantenwissenschaften},
  doi={https://doi.org/10.22331/q-2024-02-26-1266}
}

@article{endo2019,
  title = {Mitigating algorithmic errors in a Hamiltonian simulation},
  author = {Endo, Suguru and Zhao, Qi and Li, Ying and Benjamin, Simon and Yuan, Xiao},
  journal = {Phys. Rev. A},
  volume = {99},
  issue = {1},
  pages = {012334},
  numpages = {8},
  year = {2019},
  month = {1},
  publisher = {American Physical Society},
  doi = {10.1103/PhysRevA.99.012334},
  url = {https://link.aps.org/doi/10.1103/PhysRevA.99.012334}
}

@article{vazquez2023well,
  title={Well-conditioned multi-product formulas for hardware-friendly Hamiltonian simulation},
  author={Vazquez, Almudena Carrera and Egger, Daniel J and Ochsner, David and Woerner, Stefan},
  journal={Quantum},
  volume={7},
  pages={1067},
  year={2023},
  publisher={Verein zur F{\"o}rderung des Open Access Publizierens in den Quantenwissenschaften},
  doi={https://doi.org/10.22331/q-2023-07-25-1067}
}

@article{aftab2024multi,
  title={Multi-product Hamiltonian simulation with explicit commutator scaling},
  author={Aftab, Junaid and An, Dong and Trivisa, Konstantina},
  journal={arXiv preprint arXiv:2403.08922},
  year={2024},
  doi={
https://doi.org/10.48550/arXiv.2403.08922
}
}

@article{low2019well,
  title={Well-conditioned multiproduct Hamiltonian simulation},
  author={Low, Guang Hao and Kliuchnikov, Vadym and Wiebe, Nathan},
  journal={arXiv preprint arXiv:1907.11679},
  year={2019},
  doi={
https://doi.org/10.48550/arXiv.1907.11679}
}

@article{childs2021theory,
  title={Theory of trotter error with commutator scaling},
  author={Childs, Andrew M and Su, Yuan and Tran, Minh C and Wiebe, Nathan and Zhu, Shuchen},
  journal={Physical Review X},
  volume={11},
  number={1},
  pages={011020},
  year={2021},
  publisher={APS},
  doi={
https://doi.org/10.1103/PhysRevX.11.011020}
}

@article{tran2021faster,
  title={Faster digital quantum simulation by symmetry protection},
  author={Tran, Minh C and Su, Yuan and Carney, Daniel and Taylor, Jacob M},
  journal={PRX Quantum},
  volume={2},
  number={1},
  pages={010323},
  year={2021},
  publisher={APS},
  doi={https://doi.org/10.1103/PRXQuantum.2.010323}
}

@article{csahinouglu2021hamiltonian,
  title={Hamiltonian simulation in the low-energy subspace},
  author={{\c{S}}ahino{\u{g}}lu, Burak and Somma, Rolando D},
  journal={npj Quantum Information},
  volume={7},
  number={1},
  pages={119},
  year={2021},
  publisher={Nature Publishing Group UK London},
  doi={https://doi.org/10.1038/s41534-021-00451-w}
}

@article{childs2012hamiltonian,
author = {Childs, Andrew M. and Wiebe, Nathan},
title = {Hamiltonian simulation using linear combinations of unitary operations},
year = {2012},
issue_date = {November 2012},
publisher = {Rinton Press, Incorporated},
address = {Paramus, NJ},
volume = {12},
number = {11–12},
issn = {1533-7146},
journal = {Quantum Info. Comput.},
pages = {901–924},
numpages = {24},
doi={
https://doi.org/10.26421/QIC12.11-12}
}

@article{childs2018toward,
  title={Toward the first quantum simulation with quantum speedup},
  author={Childs, Andrew M and Maslov, Dmitri and Nam, Yunseong and Ross, Neil J and Su, Yuan},
  journal={Proceedings of the National Academy of Sciences},
  volume={115},
  number={38},
  pages={9456--9461},
  year={2018},
  publisher={National Acad Sciences},
  doi={https://doi.org/10.1073/pnas.1801723115}
}

@article{lloyd1996universal,
  title={Universal quantum simulators},
  author={Lloyd, Seth},
  journal={Science},
  volume={273},
  number={5278},
  pages={1073--1078},
  year={1996},
  publisher={American Association for the Advancement of Science},
  doi={10.1126/science.273.5278.1073}
}

@article{faehrmann2022randomizing,
  title={Randomizing multi-product formulas for Hamiltonian simulation},
  author={Faehrmann, Paul K and Steudtner, Mark and Kueng, Richard and Kieferova, Maria and Eisert, Jens},
  journal={Quantum},
  volume={6},
  pages={806},
  year={2022},
  publisher={Verein zur F{\"o}rderung des Open Access Publizierens in den Quantenwissenschaften},
  doi={https://doi.org/10.22331/q-2022-09-19-806}
}

@article{zhuk2023trotter,
  title={Trotter error bounds and dynamic multi-product formulas for Hamiltonian simulation},
  author={Zhuk, Sergiy and Robertson, Niall and Bravyi, Sergey},
  journal={arXiv preprint arXiv:2306.12569},
  year={2023},
doi = {10.48550/arXiv.2306.12569}
}

@article{tran2020destructive,
  title={Destructive error interference in product-formula lattice simulation},
  author={Tran, Minh C and Chu, Su-Kuan and Su, Yuan and Childs, Andrew M and Gorshkov, Alexey V},
  journal={Physical review letters},
  volume={124},
  number={22},
  pages={220502},
  year={2020},
  publisher={APS},
  doi={https://doi.org/10.1103/PhysRevLett.124.220502}
}

@article{roggero2020quantum,
  title={Quantum computing for neutrino-nucleus scattering},
  author={Roggero, Alessandro and Li, Andy CY and Carlson, Joseph and Gupta, Rajan and Perdue, Gabriel N},
  journal={Physical Review D},
  volume={101},
  number={7},
  pages={074038},
  year={2020},
  publisher={APS},
  doi={https://doi.org/10.1103/PhysRevD.101.074038}
}

@article{clinton2024towards,
  title={Towards near-term quantum simulation of materials},
  author={Clinton, Laura and Cubitt, Toby and Flynn, Brian and Gambetta, Filippo Maria and Klassen, Joel and Montanaro, Ashley and Piddock, Stephen and Santos, Raul A and Sheridan, Evan},
  journal={Nature Communications},
  volume={15},
  number={1},
  pages={211},
  year={2024},
  publisher={Nature Publishing Group UK London},
  doi={https://doi.org/10.1038/s41467-023-43479-6}
}

@article{shaw2020quantum,
  title={Quantum algorithms for simulating the lattice schwinger model},
  author={Shaw, Alexander F and Lougovski, Pavel and Stryker, Jesse R and Wiebe, Nathan},
  journal={Quantum},
  volume={4},
  pages={306},
  year={2020},
  publisher={Verein zur F{\"o}rderung des Open Access Publizierens in den Quantenwissenschaften},
  doi={https://doi.org/10.22331/q-2020-08-10-306}
}

@article{cao2019quantum,
  title={Quantum chemistry in the age of quantum computing},
  author={Cao, Yudong and Romero, Jonathan and Olson, Jonathan P and Degroote, Matthias and Johnson, Peter D and Kieferov{\'a}, M{\'a}ria and Kivlichan, Ian D and Menke, Tim and Peropadre, Borja and Sawaya, Nicolas PD and others},
  journal={Chemical reviews},
  volume={119},
  number={19},
  pages={10856--10915},
  year={2019},
  publisher={ACS Publications},
  doi={https://doi.org/10.1021/acs.chemrev.8b00803}
}

@article{childs2019faster,
  title={Faster quantum simulation by randomization},
  author={Childs, Andrew M and Ostrander, Aaron and Su, Yuan},
  journal={Quantum},
  volume={3},
  pages={182},
  year={2019},
  publisher={Verein zur F{\"o}rderung des Open Access Publizierens in den Quantenwissenschaften},
  doi={
https://doi.org/10.22331/q-2019-09-02-182}
}

@article{hagan2023composite,
  title={Composite quantum simulations},
  author={Hagan, Matthew and Wiebe, Nathan},
  journal={Quantum},
  volume={7},
  pages={1181},
  year={2023},
  publisher={Verein zur F{\"o}rderung des Open Access Publizierens in den Quantenwissenschaften},
  doi={https://doi.org/10.22331/q-2023-11-14-1181}
}

@article{cook2024parametric,
  title={Parametric Matrix Models},
  author={Cook, Patrick and Jammooa, Danny and Hjorth-Jensen, Morten and Lee, Daniel D and Lee, Dean},
  journal={arXiv preprint arXiv:2401.11694},
  year={2024},
  doi={
https://doi.org/10.48550/arXiv.2401.11694
}
}

@article{zlokapa2024hamiltonian,
  doi = {10.22331/q-2024-08-27-1449},
  url = {https://doi.org/10.22331/q-2024-08-27-1449},
  title = {Hamiltonian simulation for low-energy states with optimal time dependence},
  author = {Zlokapa, Alexander and Somma, Rolando D.},
  journal = {{Quantum}},
  issn = {2521-327X},
  publisher = {{Verein zur F{\"{o}}rderung des Open Access Publizierens in den Quantenwissenschaften}},
  volume = {8},
  pages = {1449},
  month = aug,
  year = {2024}
}

@article{sharma2024hamiltonian,
  title={Hamiltonian Simulation in the Interaction Picture Using the Magnus Expansion},
  author={Sharma, Kunal and Tran, Minh C},
  journal={arXiv preprint arXiv:2404.02966},
  year={2024},
  doi={https://doi.org/10.48550/arXiv.2404.02966}
}

@article{watson2023quantum,
  title={Quantum Algorithms for Simulating Nuclear Effective Field Theories},
  author={Watson, James D and Bringewatt, Jacob and Shaw, Alexander F and Childs, Andrew M and Gorshkov, Alexey V and Davoudi, Zohreh},
  journal={arXiv preprint arXiv:2312.05344},
  year={2023},
  doi={
https://doi.org/10.48550/arXiv.2312.05344
}
}

@article{low2019hamiltonian,
  title={Hamiltonian simulation by qubitization},
  author={Low, Guang Hao and Chuang, Isaac L},
  journal={Quantum},
  volume={3},
  pages={163},
  year={2019},
  publisher={Verein zur F{\"o}rderung des Open Access Publizierens in den Quantenwissenschaften},
  doi={	https://doi.org/10.22331/q-2019-07-12-163}
}

@article{zhao2024entanglement,
  title={Entanglement accelerates quantum simulation},
  author={Zhao, Qi and Zhou, You and Childs, Andrew M},
  journal={arXiv preprint arXiv:2406.02379},
  year={2024},
  doi={https://doi.org/10.48550/arXiv.2406.02379}
}

@article{zhao2022hamiltonian,
  title={Hamiltonian simulation with random inputs},
  author={Zhao, Qi and Zhou, You and Shaw, Alexander F and Li, Tongyang and Childs, Andrew M},
  journal={Physical Review Letters},
  volume={129},
  number={27},
  pages={270502},
  year={2022},
  publisher={APS},
  doi={https://doi.org/10.1103/PhysRevLett.129.270502}
}

@article{morales2022greatly,
  title={Greatly improved higher-order product formulae for quantum simulation},
  author={Morales, Mauro ES and Costa, Pedro and Burgarth, Daniel K and Sanders, Yuval R and Berry, Dominic W},
  journal={arXiv preprint arXiv:2210.15817},
  year={2022},
  doi={https://doi.org/10.48550/arXiv.2210.15817}
}

@article{bosse2024efficient,
  title={Efficient and practical Hamiltonian simulation from time-dependent product formulas},
  author={Bosse, Jan Lukas and Childs, Andrew M and Derby, Charles and Gambetta, Filippo Maria and Montanaro, Ashley and Santos, Raul A},
  journal={arXiv preprint arXiv:2403.08729},
  year={2024},
  doi={https://doi.org/10.48550/arXiv.2403.08729}
}

@article{campbell2019random,
  title={Random compiler for fast Hamiltonian simulation},
  author={Campbell, Earl},
  journal={Physical review letters},
  volume={123},
  number={7},
  pages={070503},
  year={2019},
  publisher={APS},
  doi={
https://doi.org/10.1103/PhysRevLett.123.070503}
}

@article{nakaji2023qswift,
  title = {High-Order Randomized Compiler for Hamiltonian Simulation},
  author = {Nakaji, Kouhei and Bagherimehrab, Mohsen and Aspuru-Guzik, Al\'an},
  journal = {PRX Quantum},
  volume = {5},
  issue = {2},
  pages = {020330},
  numpages = {23},
  year = {2024},
  publisher = {American Physical Society},
  doi = {10.1103/PRXQuantum.5.020330},
  url = {https://link.aps.org/doi/10.1103/PRXQuantum.5.020330}
}

@article{babbush2015chemical,
  title={Chemical basis of Trotter-Suzuki errors in quantum chemistry simulation},
  author={Babbush, Ryan and McClean, Jarrod and Wecker, Dave and Aspuru-Guzik, Al{\'a}n and Wiebe, Nathan},
  journal={Physical Review A},
  volume={91},
  number={2},
  pages={022311},
  year={2015},
  publisher={APS},
  doi={https://doi.org/10.1103/PhysRevA.91.022311

}
}

@article{richardson1911approximate,
  title={The approximate arithmetical solution by finite differences of physical problems involving differential equations, with an application to the stresses in a masonry dam},
  author={Richardson, Lewis Fry },
  journal={Philosophical Transactions of the Royal Society A},
  volume={210},
  number={9},
  pages={307-357},
  year={1911},
  doi={https://doi.org/10.1098/rsta.1911.0007}
}

@article{yuan2019theory,
  title={Theory of variational quantum simulation},
  author={Yuan, Xiao and Endo, Suguru and Zhao, Qi and Li, Ying and Benjamin, Simon C},
  journal={Quantum},
  volume={3},
  pages={191},
  year={2019},
  publisher={Verein zur F{\"o}rderung des Open Access Publizierens in den Quantenwissenschaften},
  doi={https://doi.org/10.22331/q-2019-10-07-191}
}

@article{chen2024adaptive,
  title={Adaptive variational simulation for open quantum systems},
  author={Chen, Huo and Gomes, Niladri and Niu, Siyuan and de Jong, Wibe Albert},
  journal={Quantum},
  volume={8},
  pages={1252},
  year={2024},
  publisher={Verein zur F{\"o}rderung des Open Access Publizierens in den Quantenwissenschaften},
  doi={https://doi.org/10.22331/q-2024-02-13-1252}
}

@article{rendon2024towards,
  title={Towards Dequantizing Quantum Signal Processing},
  author={Rendon, Gumaro},
  journal={arXiv preprint arXiv:2311.01533},
  year={2024},
  doi={
https://doi.org/10.48550/arXiv.2311.01533}
}

@book{higham2008functions,
author = {Higham, Nicholas J.},
title = {Functions of Matrices},
publisher = {Society for Industrial and Applied Mathematics},
year = {2008},
doi = {10.1137/1.9780898717778},
address = {},
edition   = {},
URL = {https://epubs.siam.org/doi/abs/10.1137/1.9780898717778},
eprint = {https://epubs.siam.org/doi/pdf/10.1137/1.9780898717778}
}

@article{jordan2012quantum,
  title={Quantum algorithms for quantum field theories},
  author={Jordan, Stephen P and Lee, Keith SM and Preskill, John},
  journal={Science},
  volume={336},
  number={6085},
  pages={1130--1133},
  year={2012},
  publisher={American Association for the Advancement of Science},
  doi={10.1126/science.1217069}
}

@article{watson2024exponentially,
  title={Exponentially Reduced Circuit Depths Using Trotter Error Mitigation},
  author={Watson, James D and Watkins, Jacob},
  journal={arXiv preprint arXiv:2408.14385},
  year={2024},
  doi={
https://doi.org/10.48550/arXiv.2408.14385}
}

@article{gong2023improved,
  title={Improved Digital Quantum Simulation by Non-Unitary Channels},
  author={Gong, Weiyuan and Kharkov, Yaroslav and Tran, Minh C and Bienias, Przemyslaw and Gorshkov, Alexey V},
  journal={arXiv preprint arXiv:2307.13028},
  year={2023},
  doi={
https://doi.org/10.48550/arXiv.2307.13028
}
}

@article{cho2024doubling,
  title={Doubling the order of approximation via the randomized product formula},
  author={Cho, Chien-Hung and Berry, Dominic W and Hsieh, Min-Hsiu},
  journal={Physical Review A},
  volume={109},
  number={6},
  pages={062431},
  year={2024},
  publisher={APS},
  doi={
https://doi.org/10.1103/PhysRevA.109.062431}
}

@article{wan2022randomized,
  title={Randomized quantum algorithm for statistical phase estimation},
  author={Wan, Kianna and Berta, Mario and Campbell, Earl T},
  journal={Physical Review Letters},
  volume={129},
  number={3},
  pages={030503},
  year={2022},
  publisher={APS},
  doi={https://doi.org/10.1103/PhysRevLett.129.030503}
}

@article{martyn2024halving,
  title={Halving the Cost of Quantum Algorithms with Randomization},
  author={Martyn, John M and Rall, Patrick},
  journal={arXiv preprint arXiv:2409.03744},
  year={2024},
  doi={https://doi.org/10.48550/arXiv.2409.03744}
}

@article{wang2024faster,
  title={Faster Quantum Algorithms with ``Fractional''-Truncated Series},
  author={Wang, Yue and Zhao, Qi},
  journal={arXiv preprint arXiv:2402.05595},
  year={2024},
  doi={https://doi.org/10.48550/arXiv.2402.05595}
}

@article{ouyang2020compilation,
  title={Compilation by stochastic Hamiltonian sparsification},
  author={Ouyang, Yingkai and White, David R and Campbell, Earl T},
  journal={Quantum},
  volume={4},
  pages={235},
  year={2020},
  publisher={Verein zur F{\"o}rderung des Open Access Publizierens in den Quantenwissenschaften},
  doi={https://doi.org/10.22331/q-2020-02-27-235}
}

@article{bagherimehrab2024faster,
  title={Faster Algorithmic Quantum and Classical Simulations by Corrected Product Formulas},
  author={Bagherimehrab, Mohsen and Berry, Dominic W and Schleich, Philipp and Aldossary, Abdulrahman and Angulo, Jorge A and Aspuru-Guzik, Alan},
  journal={arXiv preprint arXiv:2409.08265},
  year={2024},
  doi={https://doi.org/10.48550/arXiv.2409.08265}
}

@article{poulin2014trotter,
  title={The Trotter step size required for accurate quantum simulation of quantum chemistry},
  author={Poulin, David and Hastings, Matthew B and Wecker, Dave and Wiebe, Nathan and Doherty, Andrew C and Troyer, Matthias},
  journal={arXiv preprint arXiv:1406.4920},
  year={2014},
  doi={https://doi.org/10.48550/arXiv.1406.4920}
}

@article{cao2024marqsim,
  title={MarQSim: Reconciling Determinism and Randomness in Compiler Optimization for Quantum Simulation},
  author={Cao, Xiuqi and Zhou, Junyu and Liu, Yuhao and Shi, Yunong and Li, Gushu},
  journal={arXiv preprint arXiv:2408.03429},
  year={2024},
  doi={https://doi.org/10.48550/arXiv.2408.03429}
}

@article{kang2023leveraging,
  title={Leveraging hamiltonian simulation techniques to compile operations on bosonic devices},
  author={Kang, Christopher and Soley, Micheline B and Crane, Eleanor and Girvin, SM and Wiebe, Nathan},
  journal={arXiv preprint arXiv:2303.15542},
  year={2023},
  doi={https://doi.org/10.48550/arXiv.2303.15542}
}

@article{mckeever2023classically,
  title={Classically optimized Hamiltonian simulation},
  author={Mc Keever, Conor and Lubasch, Michael},
  journal={Physical review research},
  volume={5},
  number={2},
  pages={023146},
  year={2023},
  publisher={APS},
  doi={https://doi.org/10.1103/PhysRevResearch.5.023146}
}

@article{granet2024hamiltonian,
  title={Hamiltonian dynamics on digital quantum computers without discretization error},
  author={Granet, Etienne and Dreyer, Henrik},
  journal={npj Quantum Information},
  volume={10},
  number={1},
  pages={82},
  year={2024},
  publisher={Nature Publishing Group UK London},
  doi={https://doi.org/10.1038/s41534-024-00877-y}
}

@article{kiss2023importance,
  title={Importance sampling for stochastic quantum simulations},
  author={Kiss, Oriel and Grossi, Michele and Roggero, Alessandro},
  journal={Quantum},
  volume={7},
  pages={977},
  year={2023},
  publisher={Verein zur F{\"o}rderung des Open Access Publizierens in den Quantenwissenschaften},
  doi={https://doi.org/10.22331/q-2023-04-13-977}
}

@article{chakraborty2024implementing,
  title={Implementing any linear combination of unitaries on intermediate-term quantum computers},
  author={Chakraborty, Shantanav},
  journal={Quantum},
  volume={8},
  pages={1496},
  year={2024},
  publisher={Verein zur F{\"o}rderung des Open Access Publizierens in den Quantenwissenschaften},
  doi={https://doi.org/10.22331/q-2024-10-10-1496}
}
		\endgroup}

\appendix

\section{Existence of the Logarithm of the qDRIFT Channel}
\label{Sec:Existence_of_Logarithm}

A subtle point is that the logarithm for the qDRIFT channel $\cE$ may not generally exist, where 
\begin{align}
    \cE = \sum_j p_j e^{-it\lambda \ad_{H_j}}
\end{align}
where $p_j=h_j/\lambda$ and $\lambda = \sum_j h_j$ is the qDRIFT channel.
Indeed, for general Hamiltonians and timescales $t$ it does not\footnote{Consider the single-qubit Hamiltonian defined by local terms $I,X,Y,Z$ and $t=\pi/2$. In this case $\cE$ is the totally depolarizing channel and maps all quantum states to the maximally mixed state, and hence is not invertible. Since it is not invertible, its logarithm cannot exist.}.
In this section we show that provided the timescale is $t$ is small enough, then $\log(\cE)$ exists.
\begin{theorem}[Theorem 1.27 of \cite{higham2008functions}] \label{Theorem:Log_Existence}
    Let $A\in \mathbb{C}^{n\times n}$ be a complex square matrix.
    The matrix logarithm $\log(A)$ exists provided $A$ has trivial kernel. 
\end{theorem}
\noindent Here we prove that $\cE$ has non-trivial kernel provided $t$ is small enough.
\begin{lemma}
    The matrix logarithm $\log \cE$ exists provided $t<1/(2\lambda)$, where $\lambda = \sum_j h_j$.
\end{lemma}
\begin{proof}
Let $B$ be a matrix $B\in \mathbb{C}^{n\times n}$.
Using the variation-of-parameters formula, \cref{Lemma:Variation_of_Parameters}, we can write
\begin{align}
    e^{-it\lambda \ad_{H_j}}B   = B + (-it\lambda )\int_0^1 dx \ \  \ad_{H_j}e^{-it\lambda x\ad_{H_j}}B. 
\end{align}
Then if we expand $\cE$ we see that:
\begin{align}
    \cE(B) = \sum_j p_je^{-it\lambda \ad_{H_j}}B = B - it\lambda \sum_j p_j \int_0^1 dx \ \  \ad_{H_j}  e^{-it\lambda x\ad_{H_j}}(B)
\end{align}
We now proceed with the aim of proving a contradiction.
Assume that $B$ is in the kernel of $\cE$, then we have that
\begin{align}
    B = it\lambda \sum_j p_j \int_0^1 dx \  \ad_{H_j}e^{-it\lambda x\ad_{H_j}} B 
\end{align}
We take the $1$-norm of the right-hand side:
\begin{align}
    \norm{it\sum_j p_j \int_0^1 dx \ \ad_{H_j} e^{-itx\lambda \ad_{H_j}} B} 
    &\leq \left( \sum_k p_k \right)\left( 2t \lambda \norm{B}_1\max_j \norm{H_j}  \right) \\
    &\leq 2t\lambda  \norm{B}_1  \max_j \norm{H_j} \\
    &\leq 2t\lambda  \norm{B}_1  ,
\end{align}
where we have used that $\norm{H_j}\leq 1$.
Now choose $2t \lambda  <1$. 
This would imply that
\begin{align}
    \norm{B}_1 &<  2t\lambda  \norm{B}_1.
\end{align}
Here we have used $t>0$ by assumption, and $\norm{B}_1>0$ unless $B=\mathbf{0}_{n\times n},$ (by the properties of proper norms) then we have a contradiction.
Thus, assuming $B\neq \mathbf{0}_{n\times n},$ then $\cE$ must have trivial kernel when $t< 1/(2\lambda)$ by \cref{Theorem:Log_Existence}.
\end{proof}

\section{Bounding the Series Expansion for $\cG(s)$}
First we find an explicit series expansion for the $\cG(s)$ function, where from \cref{Eq:E_j_Definition}, we remember that 
\begin{align}
    \cG(s) = \ad_H + \sum_{j=1}^\infty (sT)^jE_{j+1},
\end{align}
and that $E_j$ has absorbed a factor of $\lambda^j$.
We can then bound the coefficients in the following lemma.
\begin{lemma}
\label{Lemma:Ek_Bound}
    The coefficients $E_j$ satisfy the following bound
    \begin{align}
        \dnorm{E_j}\leq (4\lambda )^j
    \end{align}
\end{lemma}
\begin{proof}
    From \cref{Eq:cG_Power_Series} we can write
    \begin{align}
        \cG(s) &= \ad_H + \sum_{j=1}^\infty E_{j+1}s^jT^j \\
        &= \frac{1}{s} \left( (\cI-\cE_s) + \sum_{k=2}^\infty \frac{(-1)^{k}}{k}(\cI-\cE_s)^k \right). 
    \end{align}
    Writing 
    \begin{align}
        \cE_s &= \cI + sT\ad_H + \sum_j p_j \left( 
\sum_{k=2} \frac{(\lambda Ts)^k}{k!} \ad^{k}_{H_j}  \right)      \\
        &= \cI+ sT \ad_H + \sum_{j=2}^\infty B_j s^j,
    \end{align}
    where we have just collected the terms into $B_j$.
    We then realise that the $j^{th}$ term in the series of $\cE_s$ is bounded by $\dnorm{B_j}\leq   \frac{(\lambda T)^j }{j!}\max_\ell \dnorm{\ad^k_{H_\ell}}\sum_j p_j \leq \frac{(2\lambda T)^j}{j!}$.
    Grouping together all of the terms of order $O(s^k)$ and taking the triangle inequality, we get 
    \begin{align}
        \dnorm{E_{k+1}}s^kT^k &\leq s^kT^k\sum_{l=1}^k \sum_{\substack{j_1,\dots, j_l \\ j_1+j_2+\dots j_1=k }}\prod_{j_k}  \dnorm{B_{j_k}} \\
        &\leq s^k \sum_{l=1}^k \sum_{\substack{j_1,\dots, j_l \\ j_1+j_2+\dots j_1=k }}\prod_{k=1}^l  \frac{(2\lambda T)^{j_k}}{j_k!} \\
        &\leq s^k \sum_{l=1}^k (2\lambda T)^k \sum_{\substack{j_1,\dots, j_l \\ j_1+j_2+\dots j_1=k }} \prod \frac{1}{j_k!} \\
        &\leq s^k  \sum_{l=1}^k \frac{(2\lambda T)^k}{k!} \sum_{\substack{j_1,\dots, j_l \\ j_1+j_2+\dots j_1=k }} \frac{k!}{j_1! \dots j_k!}\\
        &\leq s^k \sum_{l=1}^k \frac{(2\lambda T)^k}{k!} l^k  
    \end{align}
    where we have used a bound on the multinomial coefficient being less that $\leq l^k$.
    Now use that $\sum_{l=1}^k l^k\leq \int_0^k \ l^k \  dl = \frac{k^{k+1}}{k+1}\leq k^k$.
    Then using that $2^k>\frac{k^k}{k!}$ for $l\leq k$, we see that we get the bound
    \begin{align}
        \dnorm{E_k} &\leq  (4\lambda)^k 
    \end{align}
\end{proof}

\end{document}